%% file: main.tex
\pdfoutput=1
\documentclass{article}
\usepackage{iclr2027_preprint,times}
\iclrfinalcopy
\usepackage[T1]{fontenc}
\usepackage[utf8]{inputenc}
\usepackage{microtype}
\usepackage{amsmath,amssymb,amsthm}
\newtheoremstyle{compact}{3pt}{3pt}{\itshape}{}{\bfseries}{.}{.5em}{}
\theoremstyle{compact}
\newtheorem{proposition}{Proposition}
\newtheorem{lemma}{Lemma}
\usepackage{graphicx,booktabs,tabularx,array}
\usepackage{placeins}
\usepackage{xcolor}
\usepackage{hyperref}
\usepackage{url}
\definecolor{linkblue}{HTML}{285F8D}
\hypersetup{colorlinks=true,linkcolor=linkblue,citecolor=linkblue,urlcolor=linkblue,
 pdftitle={EfficientAgent: What Makes KV Cache Offloading Work for Concurrent Agents?},pdfauthor={Kunming Shao, Jierun Chen, Jiangnan Yu, Xiao-Hui Li, Chaofan Tao, Yanli Wang, Huanxin Lin, Kwang-Ting Cheng, Chi Ying Tsui, Haoli Bai}}
\newcommand{\method}{\textsc{EfficientAgent}}
\newcommand{\qthree}{Qwen3-Coder-30B-A3B-Instruct}
\newcolumntype{Y}{>{\raggedright\arraybackslash}X}
\title{EfficientAgent: What Makes KV Cache \\ Offloading Work for Concurrent Agents?}
\makeatletter
\let\preprint@maketitle\@maketitle
\renewcommand{\@maketitle}{\preprint@maketitle\lhead{Preprint}}
\makeatother
\author{Kunming Shao$^{1}$\quad Jierun Chen$^{2}$\quad Jiangnan Yu$^{1}$\quad Xiao-Hui Li$^{2}$\quad Chaofan Tao$^{3}$\\
\textbf{Yanli Wang$^{4}$\quad Huanxin Lin$^{2}$\quad Kwang-Ting Cheng$^{1}$\quad Chi Ying Tsui$^{1}$\quad Haoli Bai$^{2,\dagger}$}\\
\normalfont $^{1}$The Hong Kong University of Science and Technology\quad $^{2}$Huawei Technologies Ltd.\\
\normalfont $^{3}$The University of Hong Kong\quad $^{4}$Sun Yat-sen University\\
\normalfont $^{\dagger}$Corresponding author}
\begin{document}
\maketitle

\begin{abstract}
LLM agents resend their whole conversation on every turn, and most of it was already processed on the previous turn. Serving systems avoid recomputing it by caching its key--value (KV) state and, when GPU memory runs out, by offloading that state to host memory. For agents, offloading gives inconsistent results: on the same coding-agent workload it speeds up one deployment, slows down another, and changes nothing on a third, even where loading a token back is several times cheaper than recomputing it. The reason is that cached state must survive until it is used again. While one agent waits for its tool, the server processes the contexts of all other agents, so an agent's prefix is reused only if the host tier holds the reusable context of the whole agent pool, which we call the reuse working set. A smaller tier keeps writing state that is evicted before anyone reads it. We present \method{}, which sizes and manages the host tier by this working set. A stack-distance model estimates the working set from agent histories to size the host tier; its predictions, made before the experiments, located the capacity at which offloading starts to pay. When the tier is too small, a runtime policy stops writing large refills of evicted context and keeps extending prefixes that are still cached; when the tier is large enough, it writes everything. On SWE-bench Verified coding agents, a host tier sized to the estimated working set cuts recomputed prompt tokens by 93\% and end-to-end time by 39\%. With a small fixed tier, the policy cuts recomputation by 35\%; with a large tier, it avoids the 4.3-fold increase caused by always filtering writes. Across three GPU types and two models, offloading pays off when the GPU has little compute per byte of host bandwidth and the host tier holds the working set. Code is available at \url{https://github.com/KunmingSHAO/efficientagent_release}.
\end{abstract}

\section{Introduction}
\label{sec:intro}

An LLM coding agent works in a loop: the model reads the task, calls a tool such as a test run, reads the result, and calls the model again \citep{wang2024openhands,yang2024sweagent}. Every call resends the whole conversation so far. In traces of coding agents solving SWE-bench Verified, a benchmark of real GitHub issues \citep{jimenez2024swebench}, 98.1\% of all prompt tokens had already been processed in the same task's previous call. Reusing this work is the main way to serve agents cheaply: prefix caching, which keeps the key--value (KV) state of processed tokens and reuses it when a later prompt starts with the same tokens, alone makes coding-agent runs up to 11.1$\times$ faster (Appendix~\ref{app:hardware}). The cost of agent inference also limits ML research itself, since reinforcement-learning rollouts, evaluation, and test-time scaling run many agents at once \citep{luo2025deepswe,cao2025skyrlagent,pan2025swegym}.

\begin{figure}[t]
 \centering
 \includegraphics[width=\linewidth]{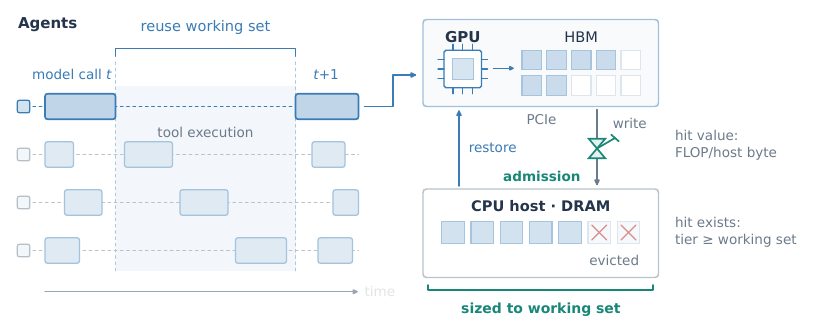}
 \caption{Overview of our \textbf{\method{}.} Left: between two turns of one agent ($t$ and $t{+}1$), the server processes the other agents' calls; this reuse working set decides whether the agent's cached prefix is still in the host tier when it returns. Right: KV state moves between the GPU's HBM and the CPU host tier over PCIe; a host hit is worth the recomputation it avoids per host-link byte, and it exists only if the tier holds the working set. Teal marks \method{}: it sizes the host tier to the working set and admits host writes under pressure (Section~\ref{sec:method}).}
 \label{fig:overview}
\end{figure}

When GPU memory cannot hold every agent's context, KV offloading keeps KV in a CPU-memory \emph{host tier} and loads it back when a later prompt needs it \citep{liu2025lmcache,gao2024cachedattention,yu2025pensieve}. Current systems decide per token: load from host memory whenever that is cheaper than recomputing, and store newly computed KV for later use. By this logic, offloading should always help agents. It does not. On the same coding-agent workload, offloading makes runs faster on RTX~3090, slower on H800, and no faster on H20, although on H20 loading a token is more than six times cheaper than recomputing it (Section~\ref{sec:model}). When we replay the recorded H20 agent calls and change only the size of the host tier, a tier four times larger cuts \emph{makespan}, the time until the last task finishes, by 39\%.

\textbf{Cached state must survive until it is reused.} While one agent runs its tool, the server processes the requests of all other agents, and their contexts push older state out of the cache. With a 5\,GiB host tier per GPU, waiting behind other agents stretches the gap between two turns of one agent to a median of 9.0\,s, 14 times the agent's own time between calls. Almost everything an agent writes is needed again: 97.3\% of the new KV of a call reappears in the task's next call. Yet 79.2\% of the KV chunks written to the host tier were evicted before that next call and had to be recomputed. Whether a prefix survives depends on how much other context the server processes between two uses of it. We call this amount the \emph{reuse working set}; it grows with the number of agents and their context length.

Two factors therefore decide whether offloading helps. The hardware decides how much a host hit saves, and across GPUs this follows peak compute per byte of host bandwidth. The working set decides whether there is a hit at all. The write policy should follow the second factor: if the host tier is smaller than the working set, writing less protects the prefixes the tier can keep; if the tier is larger, writing less throws away state that would have been reused.

\method{} uses the reuse working set in two ways. A stack-distance model, which counts the distinct KV processed between two uses of each chunk, estimates the working set from agent histories, predicts how much recomputation each host-tier size leaves, and tells how much host memory to provision. At run time, \method{} controls what enters the host tier: while the working set exceeds the tier, it keeps extending prefixes that are still cached and stops writing large refills of evicted context, the same load control that keeps operating-system memory from thrashing \citep{denning1980}; otherwise it writes everything.

Our contributions are threefold. \textbf{First}, a model of concurrent agent serving that separates what a host hit is worth from whether the hit happens, sizes the host tier by the reuse working set, and predicts before deployment where more host memory stops helping. \textbf{Second}, working-set-aware admission for LMCache, a KV-caching layer for the vLLM serving engine \citep{liu2025lmcache,kwon2023vllm}: a constant-time rule per request that declines refills only while the working-set estimate exceeds the tier and the tier is full and evicting, together with an LRU proposition that identifies which writes a tier can decline without losing a hit. \textbf{Third}, two tools for measuring agent serving, the cache-stable prompt length and dependency-preserving replay, and the findings they produce: (i)~a host tier sized to the working-set estimate removes 93\% of recomputed prefill, with the transition where the model places it; (ii)~always filtering writes helps below the working set and hurts above it, while working-set-aware admission keeps the gain and avoids the loss; (iii)~across GPU designs and models, offloading pays when compute per host-link byte is low and the host tier holds the working set.

\begin{figure}[t]
 \centering
 \includegraphics[width=\linewidth]{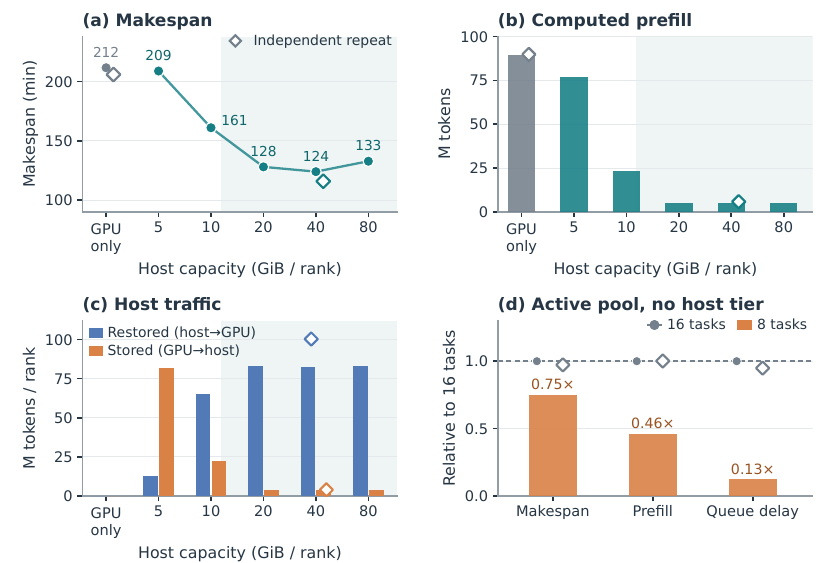}
 \caption{\textbf{A host tier provisioned at the working-set estimate cuts computed prefill by 93.1\% and makespan by 38.7\% (5 to 20\,GiB per rank).} Replay of SWE-bench Verified agent calls on eight H20 GPUs with fixed GPU KV memory; host capacity is per GPU (rank). (a)~Makespan, (b)~computed prefill, (c)~host traffic per rank; shading: capacities above the 11.4\,GiB working-set estimate (Section~\ref{sec:model}); diamonds: independent repeats. (d)~Without a host tier, halving the active pool cuts makespan, computed prefill, and queue delay.}
 \label{fig:capacity}
\end{figure}

\section{Agent Workflows and the Reuse Boundary}
\label{sec:background}

\paragraph{A closed-loop workload.}
An agent issues its next model call after the preceding response and intervening tool work have completed. With many agents, these dependencies form a closed loop: serving latency shifts the timing of later calls. We call the tasks in progress the \emph{active pool} and denote its size by $A$; the serving engine separately caps how many requests it runs at once. A cache policy acts on serving and thus on how tasks interleave; we explain makespan by three mechanisms: \emph{computed prefill}, the prompt tokens whose KV no cache tier held and the GPU therefore computed; host \emph{writes} and \emph{restores}, KV copied from GPU to host and back; and \emph{preemptions}, requests the engine suspends when GPU KV memory runs out and later resumes by recomputing their KV.

\paragraph{Exact-prefix reuse.}
For prompt $c_t$ and the previously processed sequence $z_{t-1}$, reusable state extends along their longest common prefix, subject to block alignment and residency, that is, whether a cache tier still holds the KV. We call its length, the number of prompt tokens unchanged since the previous turn, the \emph{cache-stable prompt length}; it bounds reuse. Cache keys hash the entire preceding context \citep{vllmPrefixCaching}, so rewriting early history invalidates reuse of later text even when that text is unchanged. Context folding, which replaces earlier history with a summary or a truncation, can thus reduce submitted tokens while increasing fresh prefill: on RTX~3090, history condensers that rewrite early events cut KV utilization but lowered the prefix-hit rate from 96.5\% to 51--61\% (Appendix~\ref{app:folding}). Our replay keeps every recorded prompt unchanged.

\paragraph{KV footprint and hardware balance.}
Let $L$ be the layer count, $H_{\mathrm{kv}}$ the number of KV heads, $d$ the head dimension, and $s$ bytes per element. Tensor parallelism (TP) splits each layer's KV heads across GPUs (ranks), so each GPU holds its shard of every token's KV; we therefore state KV sizes and host capacity per rank. KV bytes per token on a rank are
\begin{equation}
 \beta_{\mathrm{rank}}=2Lds\max\!\left(1,\left\lceil H_{\mathrm{kv}}/TP\right\rceil\right).
 \label{eq:footprint}
\end{equation}
For \qthree{} in BF16, $L=48$, $H_{\mathrm{kv}}=4$, and $d=128$ give 24\,KiB/token/rank at TP8 and 48\,KiB at TP2; for the dense Qwen2.5-Coder-32B-Instruct, $L=64$ and $H_{\mathrm{kv}}=8$ give 32 and 128\,KiB. Host restoration consumes CPU--GPU transfer bandwidth, while recomputation uses GPU execution capacity; H800 provides 6.7--7.0 times more peak compute per host-link byte than H20 and RTX~3090 (Section~\ref{sec:deployments}).

\section{From Offload Cost to Reuse Working Sets}
\label{sec:model}

\subsection{Valuing a recoverable prefix}

The \emph{Offload Benefit Ratio} (OBR) is the fraction of recomputation time that restoring a missing prefix of $n$ tokens saves: one means a free restore, and zero no saving. With effective prefill time $t_{\mathrm{pf}}$ per token at the operating point, effective host-to-GPU bandwidth $B_{\mathrm{H2D}}^{\mathrm{eff}}$, and a fixed restoration overhead $\tau_{\mathrm{load}}$, the times to recompute and to restore the prefix determine OBR:
\begin{align}
 T_{\mathrm{rec}}(n)&=n\,t_{\mathrm{pf}}, &
 T_{\mathrm{load}}(n)&=n\,\beta_{\mathrm{rank}}/B_{\mathrm{H2D}}^{\mathrm{eff}}+\tau_{\mathrm{load}},\nonumber\\
 \mathrm{OBR}(n)&=1-\frac{T_{\mathrm{load}}(n)}{T_{\mathrm{rec}}(n)}.
 \label{eq:obr}
\end{align}
On H20, restoring from the host tier takes a median 0.86\,$\mu$s per token per rank, whereas prefill of the model's 3.3B activated parameters needs at least 5.6\,$\mu$s per token even at the eight GPUs' peak dense BF16 throughput; OBR for long prefixes is therefore at least 0.84 (Appendix~\ref{app:model}), and at a fixed model this bound falls as peak compute per host-link byte rises (Section~\ref{sec:deployments}).

Let $U_0,U_1$ be computed prefill tokens without and with a host tier, and let $R$ be restored tokens beyond GPU-resident coverage. For the same submitted work, the serving-cost balance is
\begin{equation}
 \Delta T_{\mathrm{serve}}
 \approx -(U_0-U_1)t_{\mathrm{pf}}
       +R\,\beta_{\mathrm{rank}}/B_{\mathrm{H2D}}^{\mathrm{eff}}
       +\Delta T_{\mathrm{other}}.
 \label{eq:net}
\end{equation}
The final term includes write and restoration setup costs and exposed scheduling and overlap delays. This separates \emph{how valuable recovered computation is} from \emph{how much recovery residency enables}: a host tier that only shifts reuse from the GPU to the host adds transfer work without saving compute.

\subsection{Predicting survival across agent turns}

Each full KV chunk is identified by a content key that depends on its entire prefix. Let the \emph{reuse distance} $D(k)$ count distinct other chunks referenced between consecutive references to chunk $k$. In a fully associative LRU reference model with capacity $C$ bytes and chunk size $b$ tokens, a previously stored chunk survives when
\begin{equation}
 D(k)<\left\lfloor\frac{C}{b\beta_{\mathrm{rank}}}\right\rfloor.
 \label{eq:distance}
\end{equation}
This is the stack-distance principle of storage-hierarchy analysis \citep{mattson1970}. Prefix restoration adds a structural constraint: usable coverage is the consecutive run of available chunks from the start of the prefix. We compute host prefix coverage from chunk survival, then subtract the GPU-resident portion to estimate useful restoration. Appendix~\ref{app:model} gives the calculation.

In a backlogged pool, where each of $A$ agents has a call waiting and contexts have comparable lengths $\bar N$, the other $A-1$ agents' contexts are referenced before one agent returns, giving the working-set scale
\begin{equation}
 \widehat C_{\mathrm{reuse}}\approx(A-1)\bar N\beta_{\mathrm{rank}}.
 \label{eq:workingset}
\end{equation}
For analysis, $\bar N$ is the mean prompt length of the trace; Equation~\ref{eq:workingset} gives the scale, and the trace-based model keeps the actual ordering, sharing, and lengths.

Two ratios summarize deployment pressure: $\gamma_G=A\bar N/K_G$, the pool's context over the GPU's KV capacity of $K_G$ tokens, and $\gamma_H=\widehat C_{\mathrm{reuse}}/C_H$, the working set over the host capacity $C_H$. $\gamma_G>1$ creates an opportunity for host recovery, $\gamma_H>1$ means the host tier cannot hold what the pool reuses, and OBR values the hits that remain. On H20 (Section~\ref{sec:setup}), $A=16$ gives $\gamma_G\approx1.55$ and $\widehat C_{\mathrm{reuse}}\approx11.4$\,GiB per rank, so $\gamma_H\approx2.3$, 1.1, and 0.57 at 5, 10, and 20\,GiB. The measured capacity transition (Section~\ref{sec:capacity_results}) lies where $\gamma_H$ crosses one. GPU design enters through OBR.

\section{Working-Set-Aware Admission Scheduling}
\label{sec:method}

At run time, \method{} controls which part of the reuse working set the host tier holds by scheduling admission to the tier. The working set grows with pool size and context length (Section~\ref{sec:implications}), while host memory is fixed per server, so a tier sized for today's pool falls below the working set as pools and contexts grow; admission keeps a fixed budget effective in that regime. When the serving scheduler first looks up a waiting request's cached prefix, the runtime decides once whether to store the request's new KV in the host tier (Figure~\ref{fig:overview}b). The policy, capacity-conditioned write admission, combines a feedforward working-set estimate, which decides whether the active pool can thrash the tier, with feedback from the tier, which confirms that the tier is full and evicting.

\paragraph{Load control for the host tier.}
The policy applies the working-set principle of multiprogrammed memory \citep{denning1968,denning1980}: when the active programs' combined working set exceeds memory, the system thrashes, and load control keeps the working sets of a subset resident. A host tier below the pool's reuse working set ($\gamma_H>1$) thrashes in the same way (Figure~\ref{fig:allruns}b). Under pressure, the runtime keeps saving incremental extensions of host-resident prefixes and declines large refills, which re-store context the tier has already evicted; when pressure subsides, refills are admitted again.

\paragraph{Pressure signal and write rule.}
The runtime estimates $A$ and $\bar N$ of Equation~\ref{eq:workingset} over a recent window, from the tasks with recent requests and their prompt lengths, and the host tier reports whether it is full and evicting (Appendix~\ref{app:runtime}). A request with $n$ prompt tokens, of which the host tier already holds the first $h$, needs $u=\lfloor n/b\rfloor-\lfloor h/b\rfloor$ new full chunks. With write threshold $\kappa$, the runtime decides
\begin{equation}
 p_t=\mathbf{1}\bigl[\widehat C_{\mathrm{reuse}}>C_H\bigr]\land\mathbf{1}\bigl[\text{tier full and evicting}\bigr],
 \qquad
 \mathrm{save}(r)=\neg\bigl(p_t\land u>\kappa\bigr).
 \label{eq:save}
\end{equation}
Pressure $p_t$ thus requires both that the working-set estimate exceed the tier ($\gamma_H>1$) and that the tier be full and evicting; without a recent report, the estimate alone decides. A tier that holds the working set while evicting cold chunks stays unrestricted. A small $u$ extends a host-resident prefix, and a large $u$ refills state the tier does not hold. A skip stores none of the request's new KV and holds for the rest of the request. The rule costs one constant-time check per request, with its parameters fixed before the admission runs; any $\kappa$ from 2 to 16 declines 89.6--97.0\% of the new-chunk writes at 5\,GiB (95.3\% at $\kappa=8$; Appendix~\ref{app:runtime}). Before each copy, deduplication drops chunks the tier already holds. Offload without write admission and \emph{fixed write admission}, which applies the filter to every request with deduplication, are the endpoints $p_t\equiv0$ and $p_t\equiv1$ of this rule; capacity-conditioned admission selects between them per request.

\paragraph{Which writes a tier can decline.}
Take the LRU reference model of Equation~\ref{eq:distance} with a fixed reference stream and $K_H=\lfloor C_H/(b\beta_{\mathrm{rank}})\rfloor$ chunks, where a miss inserts its chunk unless declined.
\begin{proposition}
\label{prop:admission}
If an LRU tier of $K_H$ chunks declines insertion only at misses whose chunk is referenced next at reuse distance $D\geq K_H$, or never again, then every reference that hits under full admission also hits. Such exclusions weakly increase the hits of every chunk and the total.
\end{proposition}
A hit requires that fewer than $K_H$ distinct chunks were promoted, by a hit or an insertion, since the chunk's previous reference, and declined insertions never raise this count above $D$ (Appendix~\ref{app:runtime}). The gain comes from references whose count falls below $K_H$ once declined refills leave the stream: the working sets of the resident subset fit. Because a request's host coverage is its run of hits from the first chunk, coverage grows with the hits. Proposition~\ref{prop:admission} thus formalizes, for exact-prefix KV tiers with consecutive-prefix coverage, the principle behind cache bypassing, dead-block prediction, and re-reference interval prediction in processor caches \citep{johnson1999bypass,lai2001deadblock,khan2010sampling,jaleel2010rrip}: a block whose next reference lies beyond the cache's reach gains nothing from insertion. The runtime rule targets this set with two observable signals: the working-set estimate, which places the tier on either side of the condition, and whether the tier is full and evicting. Section~\ref{sec:policy_results} measures the effect on restores in both directions.

\section{Evaluation}
\label{sec:eval}

\subsection{Experimental setup}
\label{sec:setup}

We run OpenHands, an open-source coding-agent platform \citep{wang2024openhands}, with its CodeActAgent and \qthree{} in BF16 on vLLM~0.13.0 with prefix caching and LMCache~0.3.12, on eight H20 GPUs with TP8; the engine runs at most 16 requests at once, and the host tier stores 1,024-token chunks. Each deployment (8$\times$H20, 8$\times$RTX~3090, 2$\times$H800, and 8$\times$H800) sets TP and the share of GPU memory the engine may use so that serving is memory-constrained (Appendix Table~\ref{tab:config}); the RTX~3090 and 2$\times$H800 deployments also serve the dense Qwen2.5-Coder-32B-Instruct. On H20, the GPU KV cache holds 343K tokens in all host-capacity and policy comparisons. We vary the active pool between $A=8$ and 16 and the host capacity between 3 and 80\,GiB per rank. The live SWE-bench Verified runs, in which the agent acts on its own generations, use the same stack with $A=16$.

\paragraph{Dependency-preserving replay.}
To compare policies on identical work, we replay the SWE-bench Verified trajectories of the live H20 run without a host tier: 4,427 model calls with 147.1M prompt tokens. Model calls average 56 per task with a mean input of 33K tokens and a maximum of 238K; each task submits 1.86M prompt tokens, 80 times its 23K output tokens. Each replayed call submits its recorded prompt and reproduces every recorded output token exactly; a task's next call follows once its previous call finishes and the recorded agent and tool time elapses. Replay thus fixes each call's tokens while serving decisions change queueing and task interleaving. Runs start with a fresh server and empty cache; makespan includes all recorded agent and tool time (Appendix~\ref{app:replay}).

\paragraph{Measurements.}
\emph{Offload} denotes the LMCache host tier without write admission; it is the reference at every host budget. Figures plot every run, including independent repeats of recomputation and of 40\,GiB offload; the two recompute runs reproduce computed prefill to 0.2\%. Where two policies make the same decisions, their runs reproduce each other: at 40\,GiB, conditioned admission never engages its filter and reproduces offload (5.28M computed prefill tokens and 124.3\,min; offload 5.27M and 5.29M, 124.0 and 123.7\,min); at 5\,GiB, it engages the filter for 90\% of requests and reproduces fixed admission to 1.6\% in computed prefill and 0.4\% in makespan.

\subsection{Host capacity determines useful recovery}
\label{sec:capacity_results}

Host capacity, which the working-set model sizes, decides whether offload recovers computation (Figure~\ref{fig:capacity}). From 5 to 10 and 20\,GiB per rank, computed prefill falls from 76.7M to 23.5M and 5.3M tokens (93.1\%), and makespan follows, from 209 to 161 and 128\,min (38.7\%). The 5\,GiB tier leaves makespan unchanged from recomputation (212\,min). Across these points $\gamma_H$ falls from 2.3 through 1.1 to 0.57: the transition occurs where the host budget crosses the reuse working set.

At 5\,GiB, offload writes 6.5 tokens to the host for each token it restores (81.7M stored, 12.6M restored per rank). From 5 to 20\,GiB, host writes fall to 4.1M tokens per rank, restores rise to 82.8M, and preemptions fall from 1,808 to 46. A larger tier retains usable prefixes and avoids reconstructing and rewriting lost state; below the transition, write admission targets this write volume.

Beyond it, computed prefill levels off at 5.3--5.9M tokens from 20 to 80\,GiB. The GPU tier shows the same transition on RTX~3090 (Figure~\ref{fig:sweep}): raising the engine's share of GPU memory from 0.65 to 0.95 cuts wall-clock from 407 to 73\,min as the prefix-hit rate rises from 44.6\% to 98.1\%.

\subsection{Concurrency changes the working set}
\label{sec:concurrency}

Halving the active pool shrinks the working-set scale from 11.4 to 5.3\,GiB per rank, and the transition moves with it. Under recomputation, it cuts computed prefill from 89.7M to 41.5M tokens and makespan from 212 to 159\,min (24.8\%), with GPU memory and the engine's request cap unchanged. With $A=8$, tiers below 5\,GiB already recover prefixes: a 3\,GiB tier ($\gamma_H\approx1.8$) cuts computed prefill to 29.0M tokens and a 4.5\,GiB tier ($\gamma_H\approx1.2$) to 17.1M, with makespans of 151 and 147\,min. The active pool sets both the demand for host storage and its benefit.

\subsection{Stack-distance predictions locate the capacity transition}
\label{sec:prediction}

We recorded stack-distance predictions before running 10, 20, and 80\,GiB per rank at $A=16$ (Table~\ref{tab:prediction}) and 0, 3, and 4.5\,GiB at $A=8$ (Table~\ref{tab:prediction-n8}). The model places the capacity transition between 5 and 20\,GiB and predicts that capacity stops paying beyond it: it predicted 5.19--5.78M tokens of computed prefill at 20\,GiB and 5.19--5.31M at 80\,GiB, and the runs computed 5.28M and 5.29M. At 10\,GiB ($\gamma_H\approx1.1$), inside the transition region, the run recovers more than predicted, restoring 65.2M tokens per rank (predicted 55.3--58.9M). The same estimate that locates the transition gates admission (Section~\ref{sec:method}).

\subsection{Write admission reverses across the capacity transition}
\label{sec:policy_results}

\begin{figure}[t]
 \centering
 \includegraphics[width=\linewidth]{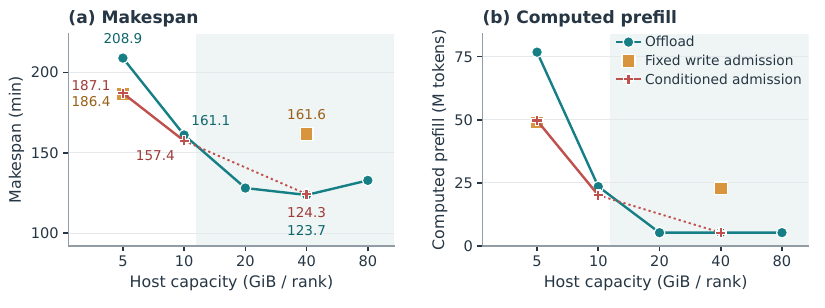}
 \caption{\textbf{Capacity-conditioned admission follows the faster policy on both sides of the capacity transition.} Replay on H20; offload has no write admission. (a)~Makespan (min), labeled at 5, 10, and 40\,GiB per rank; (b)~computed prefill. Shading: capacities above the 11.4\,GiB working-set estimate.}
 \label{fig:policy}
\end{figure}

\begin{table}[t]
 \centering\small
 \caption{\textbf{Offload pays on low-ratio GPUs once host hits survive and slows every H800 run.} Peak dense BF16 TFLOPS, host link per direction, and memory per GPU; offload/recompute for the MoE and the dense model. Bold: offload faster than recomputation.}
 \label{tab:hardware}
 \setlength{\tabcolsep}{4.5pt}
 \begin{tabular}{@{}lrlrrl@{}}
 \toprule
 GPU & BF16 TFLOPS & Host link (GB/s) & Memory (GB) & FLOP/byte & Offload/recompute \\
 \midrule
 RTX~3090 & 71 & PCIe Gen4, 32 & 24 & 2.2K & \textbf{0.91}; dense \textbf{0.92} \\
 H20 & 148 & PCIe Gen5, 64 & 96 & 2.3K & \textbf{0.60} (tier $\geq$ working set) \\
 H800 & 989.5 & PCIe Gen5, 64 & 80 & 15.5K & 1.08--1.87; dense 1.23 \\
 \bottomrule
 \end{tabular}
\end{table}

At 5\,GiB per rank, fixed write admission cuts computed prefill by 36\%, from 76.7M to 48.9M tokens, and makespan by 10.8\%, from 209 to 186\,min. Host writes fall from 81.7M to 2.6M tokens per rank: the request filter skips the writes of 1,275 requests (52,131 chunks per rank), and deduplication removes 44 chunk copies per rank.

At 40\,GiB the effect reverses. Fixed write admission raises computed prefill 4.3-fold, from 5.3M to 22.8M tokens, and lengthens makespan by 30.6\%, from 124 to 162\,min. It skips the writes of 290 requests that the larger tier could retain, and the lost recovery outweighs the reduced write volume.

Capacity-conditioned admission keeps the 5\,GiB gain and avoids the 40\,GiB loss (Figure~\ref{fig:policy}). At 5\,GiB, it engages the filter for 90.2\% of requests and skips the writes of 1,270; computed prefill falls by 35\%, from 76.7M to 49.7M tokens, and makespan by 10.4\%, from 208.9 to 187.1\,min. At 40\,GiB, the working-set estimate stays below the tier for every request, including 2,369 with a full, evicting tier, so the pressure signal never engages and every request's new KV is stored: relative to fixed admission, computed prefill falls 4.3-fold, from 22.8M to 5.3M tokens, and makespan by 23.1\%, from 161.6 to 124.3\,min.
At 10\,GiB, where the working-set estimate straddles the tier ($\gamma_H\approx1.1$), the policy engages the filter for the 39\% of requests whose estimate exceeds the tier: computed prefill falls by 14\%, from 23.5M to 20.1M tokens, and host writes by 73\%, from 22.4M to 6.0M tokens per rank.
Host restoration measures both sides of the condition of Proposition~\ref{prop:admission}. Where the working-set estimate signals pressure, declined writes cost no restores: at 5\,GiB, restored tokens per rank rise 3.3-fold, from 12.6M to 41.2M, with fixed admission and to 40.4M with conditioned admission, and at 10\,GiB from 65.2M to 66.9M. Where writes are declined without pressure, at 40\,GiB, restores fall from 81.9M to 65.2M tokens per rank and computed prefill rises 4.3-fold, as the condition predicts for declined chunks that return within capacity. The recorded reference streams confirm both sides: at 5\,GiB, 84.4\% (fixed) and 85.5\% (conditioned) of the declined chunks meet the condition of Proposition~\ref{prop:admission}, next reuse distance $D\geq K_H$ or no later reference; at 40\,GiB, where the tier holds the working set, 4.8\% of the chunks fixed admission declines do (Appendix~\ref{app:runtime}, Figure~\ref{fig:declined}).

\subsection{Makespan follows computed prefill across interventions}
\label{sec:allruns}

\begin{figure}[t]
 \centering
 \includegraphics[width=\linewidth]{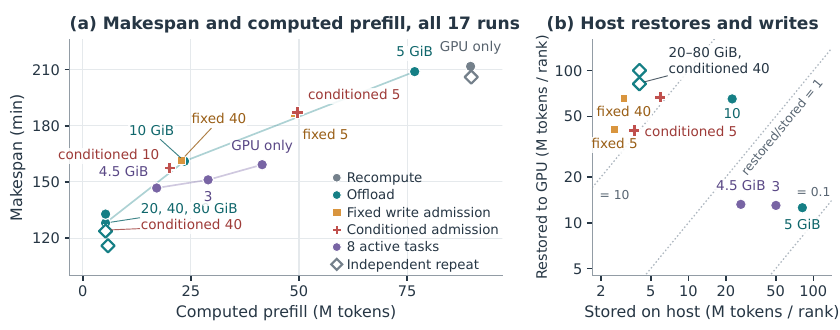}
 \caption{\textbf{With 16 active tasks, makespan follows computed prefill across capacity and write-admission runs.} Each marker is one complete run of the same replay workload; labels give host GiB per rank; diamonds mark independent repeats. (a)~Faint lines join runs without write admission in order of host capacity. (b)~Runs with a host tier; dotted lines mark restored/stored ratios of 0.1, 1, and 10.}
 \label{fig:allruns}
\end{figure}

With 16 active tasks, makespan follows computed prefill whichever intervention produced it (Figure~\ref{fig:allruns}): fixed write admission at 40\,GiB computes 22.8M tokens in 161.6\,min, and offload at 10\,GiB computes 23.5M in 161.1\,min. All 17 runs replay identical token work. The three runs with eight active tasks form a lower series. In host traffic (Figure~\ref{fig:allruns}b), offload at 5\,GiB restores 0.15 tokens per stored token, fixed write admission at 5\,GiB 16, and runs whose tier holds the working set (20--80\,GiB) 20--25.

\subsection{Hardware ratio and capacity set offload outcomes across deployments}
\label{sec:deployments}

Two factors set offload outcomes across deployments, as Section~\ref{sec:model} predicts: peak compute per host-link byte and host capacity relative to the working set (Table~\ref{tab:hardware}; Appendix~\ref{app:hardware}; deployment guide in Appendix~\ref{app:guide}). RTX~3090 and H20 are the low-ratio designs, with 2.2K and 2.3K peak FLOP per host-link byte, and offload pays on both once hits survive. It shortens the RTX~3090 runs of the MoE model (0.91) and the dense model (0.92); on H20, a tier that covers the 11.4\,GiB working set brings replay makespan to 0.60 of recomputation, and a tier below it leaves makespan unchanged (Sections~\ref{sec:intro} and~\ref{sec:capacity_results}). On H800, with 6.7--7.0 times more peak compute per host-link byte, offload lengthens every run we measured (1.08--1.87): two- and eight-GPU deployments, all three GPU memory settings, both models, and both a benchmark subset and the full benchmark. Host hits show the same split: at GPU KV usage of 80\% and above, offload lifts the median total prefix-hit rate on RTX~3090 from 33.8--34.7\% to 43.7--44.3\%; on 2$\times$H800, the rates with and without offload overlap (46.6--51.7\% and 46.9--48.8\%; Table~\ref{tab:hit-bands}).

\section{Related Work}
\label{sec:related}

\paragraph{KV reuse and tiered storage.}
PagedAttention and RadixAttention introduced paged KV allocation and automatic prefix reuse \citep{kwon2023vllm,zheng2023sglang}. FlexGen and InfiniGen offload to host memory \citep{sheng2023flexgen,lee2024infinigen}, CachedAttention and Pensieve keep multi-turn state in multi-tier KV caches \citep{gao2024cachedattention,yu2025pensieve}, HCache restores state from saved activations \citep{gao2025hcache}, and LMCache, CacheGen, and Mooncake provide KV storage and sharing, compressed streaming, and disaggregated serving \citep{liu2025lmcache,liu2024cachegen,peng2024mooncake}. \method{} runs on LMCache and sizes its host tier by the agent pool's working set.

\paragraph{Agent- and multi-turn-aware serving.}
InferCept cuts recomputation and memory waste when augmented LLMs pause for external interactions \citep{abhyankar2024infercept}. Continuum pins KV across tool calls with a time-to-live set by reload cost and queueing delay \citep{continuum2025}; MORI places programs across HBM and CPU memory by idleness, with per-tier admission control \citep{xia2026mori}; KVFlow guides eviction and prefetching by an agent step graph \citep{pan2025kvflow}; TOPAS jointly selects retained prefixes and scheduled requests \citep{topas2026}; Agentix schedules calls by program-level progress, and Preble balances prompt reuse and load \citep{luo2026agentix,srivatsa2025preble}. These systems decide what stays resident, where, or when requests run; \method{} controls whether the host tier holds the pool's reuse working set and which new state enters it, and composes with each.

\paragraph{Cache admission and working-set analysis.}
TinyLFU admits items by approximate access frequency \citep{einziger2017tinylfu}, AdaptSize tunes a size-dependent admission probability to changing request patterns \citep{berger2017adaptsize}, and Marconi admits and evicts hybrid-model prefix states by forecast reuse and compute saved per byte \citep{pan2025marconi}; these policies adapt admission to popularity and size to raise the hit ratio. Agent KV recurs in the task's next call, so a refill's value depends on whether the tier can hold the pool's working set until that call, and this is the condition our admission tests; a production study of KV-cache reuse motivates workload-aware eviction \citep{wang2025kvcachewild}. \method{} builds on the working-set and stack-distance theory of multiprogrammed memory \citep{denning1968,denning1980,mattson1970}, whose miss-ratio curves SHARDS approximates by sampling \citep{waldspurger2015shards}, and carries it to agent KV with new quantities: queue-stretched reuse distances, consecutive-prefix coverage and GPU-resident subtraction, which set the computation a host hit avoids, and the $(A-1)\bar N\beta_{\mathrm{rank}}$ scale that gates admission.

\paragraph{Context and representation changes.}
Context folding collapses completed sub-trajectories into summaries \citep{sun2025contextfolding}, SWE-Pruner prunes agent context guided by an agent-stated goal \citep{wang2026swepruner}, and KV quantization changes the stored representation \citep{liu2024kivi}. An early edit changes the key of every later chunk (Section~\ref{sec:background}), and LMCache's deployment study reports that context truncation can halve the prefix hit ratio \citep{liu2025lmcache}; the serving benefit of these methods also depends on the cache-stable length they preserve.

\paragraph{Implications for agent design and training.}
\label{sec:implications}
The working set gives agent builders a quantity they can plan with. Context-management strategies, such as folding, pruning, and summarization, are best judged by the cache-stable prompt length they preserve alongside the tokens they remove, since an early rewrite rekeys every later chunk (Section~\ref{sec:background}). The pool size and context length of a rollout or evaluation pool set the KV memory the workload needs through $(A-1)\bar N\beta_{\mathrm{rank}}$ before deployment: at 24\,KiB per token per GPU, 64 agents with 64K-token contexts need 94.5\,GiB per GPU, and 128 agents with 128K-token contexts 381\,GiB. Dependency-preserving replay turns recorded trajectories into a fixed workload on which any serving policy can be compared token for token. Together, these tools let training and evaluation pipelines size their KV tiers from the agents they run.

\section{Conclusion}

Hardware sets the value of a recovered prefix; the concurrent working set decides whether it survives until reuse. \method{} sizes and admits KV state by this working set: a stack-distance model sizes the host tier and locates the capacity transition before deployment, and admission scheduling applies load control to host writes. Across the transition, fixed write admission reverses sign, while capacity-conditioned admission keeps its gain and avoids its loss; across GPUs, offload pays where compute per host-link byte is low and the host tier covers the working set. As agents run longer in larger pools, the reuse working set becomes a first-class input for KV tiering and scheduling.

\clearpage
\section*{AI Use Statement}
AI tools assisted with manuscript editing, implementation, analysis scripts, and figure preparation. The authors reviewed the text and artifacts and take responsibility for the work.
\section*{Ethics Statement}
The experiments use existing software-maintenance benchmarks and do not release new models. More efficient coding agents can lower the cost of software development, while generated patches still require testing and security review. We report runtime and computation rather than measured energy consumption.
\section*{Reproducibility Statement}
Appendix~\ref{app:protocol} specifies the model, software, hardware, and replay protocol, and Appendix~\ref{app:replay} lists every completed run. The write-admission connector, the dependency-preserving replay framework, and the analysis tools, including the stack-distance capacity model, are available at \url{https://github.com/KunmingSHAO/efficientagent_release}.
\bibliography{references}
\bibliographystyle{iclr2027_conference}
\clearpage
\appendix
\input{appendix}
\end{document}

%% file: appendix.tex
\input{tables/measurements}
\makeatletter\setlength{\@fptop}{0pt}\makeatother %

\section{Experimental Configuration}
\label{app:protocol}

\paragraph{Agent execution.}
The H20 experiments use OpenHands v0.56.0 with CodeActAgent and \qthree{} in BF16 on SWE-bench Verified \citep{jimenez2024swebench}. Both live runs use an active pool of $A=16$ tasks (16 OpenHands workers), at most 100 agent iterations, temperature zero, top-$p=1$, a 245,760-token input limit, and a 16,384-token output limit. We use no history condenser (OpenHands' NoOp condenser), the default SWE instruction template, and benchmark hints. Each task has a separate project sandbox with two CPU cores and an 8\,GiB memory limit. Workflow time spans the first task dispatch through completion of agent inference; patch grading runs separately. Both H20 live runs use the same SWE-bench Verified instance list.

\paragraph{Serving stack and memory budgets.}
The H20 runs use vLLM~0.13.0 and LMCache~0.3.12 on eight GPUs, with the configuration in Table~\ref{tab:config}. vLLM sizes its GPU KV cache from the memory fraction listed there, the share of GPU memory the engine may use for weights, activations, and KV cache (its \texttt{gpu\_\allowbreak memory\_\allowbreak utilization} setting); on H20, the GPU KV cache stays fixed across the replay experiments. Host budgets are per TP rank: 5 and 40\,GiB per rank correspond to 40 and 320\,GiB across eight ranks. The lower part of Table~\ref{tab:config} gives the memory configuration of every deployment in Appendix~\ref{app:hardware}.

\begin{table}[ht]
 \centering\small
 \caption{Serving configuration. Upper part: the H20 stack; the active pool limits the tasks in progress, and the running-request cap limits the requests the engine executes at once. Lower part: memory configuration per deployment of \qthree{} and of the dense Qwen2.5-Coder-32B-Instruct; GPU KV capacity is in tokens, each GPU holding its shard of every token.}
 \label{tab:config}
 \begin{tabularx}{\linewidth}{@{}p{0.31\linewidth}Y@{}}
 \toprule
 Setting & Value \\
 \midrule
 Model / precision & Qwen3-Coder-30B-A3B-Instruct / BF16 \\
 Maximum model context & 262,144 tokens \\
 Running-request cap & 16 (\texttt{max\_num\_seqs}) \\
 Scheduler token budget & 8,192; chunked prefill enabled \\
 GPU prefix caching & Enabled \\
 Active pool $A$ & 16; 8 in the concurrency comparison \\
 Host KV chunk / hash & 1,024 tokens / \texttt{sha256\_cbor} \\
 KV bytes per chunk per rank & 24\,MiB \\
 Host capacity per rank & 3, 4.5, 5, 10, 20, 40, or 80\,GiB \\
 \midrule
 Deployment & Memory configuration \\
 \midrule
 $8\!\times$H20 & TP8; GPU KV capacity 343,408 tokens (7.86\,GiB per GPU) \\
 $8\!\times$RTX~3090 & TP8; memory fraction 0.65 (sweep 0.65--0.95); GPU KV capacity 345,760 tokens \\
 $2\!\times$H800, subset & TP2; memory fraction 0.65; GPU KV capacity 480,640 tokens \\
 $2\!\times$H800, full benchmark & Memory fraction 0.65 \\
 $8\!\times$H800 & Memory fraction 0.20 \\
 $8\!\times$RTX~3090, dense model & TP8; memory fraction 0.70 \\
 $2\!\times$H800, dense model & TP2; memory fraction 0.70 \\
 \bottomrule
 \end{tabularx}
\end{table}

\paragraph{Timing boundary.}
For one task, elapsed time consists of model-response intervals and the intervals between them. A model-response interval includes frontend processing, queueing, prefill, and decoding. An inter-call interval includes agent computation, tool execution, transport, and any retries in that interval. Across tasks, these activities overlap; makespan is the completion time of the last task. Consequently, summing engine times over requests or ranks does not yield makespan.

\section{Replay Construction and Complete Measurements}
\label{app:replay}

\paragraph{Workload construction.}
The fixed workload comprises the SWE-bench Verified trajectories of the live H20 run without a host tier (recomputation). It contains 4,427 successful recorded model calls, 147,124,984 prompt tokens, and 1,835,698 output tokens. Tasks enter the active pool in the recorded order. A task's next call becomes eligible after its preceding response finishes and its recorded inter-call interval elapses. The source trajectories also contain 115 calls without a usable successful completion; their elapsed time remains in the intervals between retained requests. All recorded intervals, including long waits, remain in the reported makespan.

Replay submits the exact prompt token IDs and constrains decoding, through a vLLM logits processor, to emit the recorded output tokens. Each completed run returns all 4,427 recorded sequences, every output token as recorded, without HTTP or prompt-length errors. This fixes token work and prefix identity while retaining real model execution. It measures serving performance on fixed trajectories; free-generation task quality is assessed separately in live execution. Each replay starts a fresh serving process and a cold cache. Runs with $A=8$ use the same tasks and token sequences as runs with $A=16$, with the running-request cap held at 16.

\paragraph{Repeated-prefix profile.}
Per task, the workload averages 56.0 calls (median 51, range 32--100), 1.86M prompt tokens, and 23.2K output tokens; the longest prompt has 238,105 tokens. For each call, we compute on token IDs the longest common prefix of its prompt with the same task's previous prompt, and with that previous prompt followed by its recorded output. Summed over the 4,427 calls, the first covers 142,541,707 prompt tokens (96.9\%); the second, the cache-stable prompt length of Section~\ref{sec:background}, covers 144,289,835 (98.1\%). A task's first call counts as unshared.

\paragraph{Baselines and policy variants.}
Recompute uses vLLM's GPU prefix caching and recomputes state missing from the GPU. Offload adds the LMCache host tier without write admission. Offload runs use unmodified LMCache, except one of the two independent 40\,GiB repeats, which runs through our admission runtime with both admission rules disabled (marked $^\ast$ in Table~\ref{tab:evidence-all-replay}). Fixed write admission applies the request filter of Equation~\ref{eq:save} to every request together with copy-time deduplication. Capacity-conditioned admission keeps deduplication and applies the filter under the pressure rule of Equation~\ref{eq:save}.

\paragraph{Counters and repeated runs.}
Computed prefill is the interval difference in vLLM's \texttt{request\_\allowbreak prefill\_\allowbreak kv\_\allowbreak computed\_\allowbreak tokens\_\allowbreak sum}; preemptions use \texttt{num\_\allowbreak preemptions\_\allowbreak total}. Reported host read/write volumes sum LMCache transfer-event token counts across workers and divide by eight. Mean queue delay is the corresponding queue-time sum divided by its request count. These are workload-level measurements; transfer-event counts include repeated movements of the same token. We plot each run individually; the unmarked 40\,GiB repeat additionally records host-tier telemetry (occupancy and evictions).

\EvidenceAllReplayTable

\paragraph{Capacity and concurrency.}
The 5-to-20\,GiB host-capacity change reduces computed prefill by 93.1\%, writes by 95.0\%, and makespan by 38.7\%. GPU KV capacity, active pool, and token work are fixed. Reducing the active pool from 16 to 8 tasks in the recompute configuration lowers computed prefill by 53.7\% and makespan by 24.8\% (212 to 159\,min). The concurrency experiment changes the competing history volume at a fixed running-request cap.

\section{Offload Cost and Prefix-Survival Model}
\label{app:model}

\paragraph{Hardware cost.}
OBR compares the service cost of restoring a missing prefix with recomputing that prefix at the same operating point. For restored length $n$, define
\begin{equation}
 t_{\mathrm{load}}(n)
 = \frac{T_{\mathrm{load}}(n)}{n}
 = \frac{\beta_{\mathrm{rank}}}{B_{\mathrm{H2D}}^{\mathrm{eff}}}
   + \frac{\tau_{\mathrm{load}}}{n}.
\end{equation}
As $n$ grows, the fixed overhead is amortized and $\mathrm{OBR}(n)$ approaches $1-\beta_{\mathrm{rank}}/(B_{\mathrm{H2D}}^{\mathrm{eff}}t_{\mathrm{pf}})$. Effective H2D bandwidth refers to the restore path; host writes use the opposite direction. Effective prefill cost depends on context length, batching, and model execution. Peak arithmetic throughput and host-link bandwidth bound the two costs; Section~\ref{sec:deployments} compares GPU designs by their ratio.

For H20, the restore cost comes from the LMCache retrieval records of the live offload run: 71,872 per-rank retrieval events with a median rate of 26.6\,GiB/s (27.1\,GiB/s as total bytes over total event time). At $\beta_{\mathrm{rank}}=24$\,KiB, the median corresponds to 0.86\,$\mu$s per token per rank. For prefill, a compute bound suffices. \qthree{} activates 3.3B parameters per token \citep{qwen3CoderCard}, so its linear layers alone execute 6.6\,GFLOP per prompt token; attention adds work that grows with context. At the nominal dense BF16 throughput of 148\,TFLOPS per H20, eight GPUs need at least 5.6\,$\mu$s per token. With the median restore cost, the large-$n$ limit of OBR is therefore at least $1-0.86/5.6>0.84$; attention and parallel-execution overheads make real prefill slower than this bound, which only raises OBR.

For $P$ submitted prompt tokens, let $r_g^0$ and $r_g^1$ denote exclusive GPU-hit fractions without and with offload, and $r_e$ the exclusive restored fraction. First-pass prefill is $P(1-r_g^0)$ and $P(1-r_g^1-r_e)$, respectively. If preemptions add $J_0$ and $J_1$ recomputed tokens, total computed prefill adds these quantities to the respective first-pass terms. At fixed effective costs, setting $\Delta r=r_g^1+r_e-r_g^0$ gives
\begin{equation}
 \Delta T_{\mathrm{serve}}
 \approx -[P\Delta r+J_0-J_1]\,t_{\mathrm{pf}}
          +Pr_e\,\beta_{\mathrm{rank}}/B_{\mathrm{H2D}}^{\mathrm{eff}}
          +\Delta T_{\mathrm{other}}.
\end{equation}
This is Equation~\ref{eq:net} under a common token denominator. The final term includes write costs, per-restoration setup, and changes in exposed scheduling and overlap delays. This accounting organizes service costs; completion time follows their placement along concurrent task paths.

\paragraph{Prefix coverage from a reference stream.}
We construct prefix-dependent chunk keys from the recorded token histories, so equal text after different preceding contexts has different keys. For each capacity, an LRU reference tracks stored chunks in request order. A host lookup stops at the first absent chunk: later isolated chunks do not extend the recoverable prefix. If $H_j(C)$ is this consecutive coverage for request $j$ and $G_j$ is its GPU-resident prefix, useful host restoration is
\begin{equation}
 R_j(C)=\max\{0,H_j(C)-G_j\},
\end{equation}
after chunk and block alignment. Summing the remaining uncovered prompt length estimates computed prefill. The model retains the prefix constraint, unlike a hit count that treats chunks as independent reusable objects.

We evaluate the model on the recorded orderings of the 5 and 40\,GiB reference runs; each prediction range spans these two interleavings. Table~\ref{tab:prediction} reports the predictions recorded before the 10, 20, and 80\,GiB runs together with their measurements.

\begin{table}[ht]
 \centering\small
 \caption{Prefix-survival predictions and subsequent measurements. Prefill and restored volumes are millions of tokens; restored volume is per rank. Prediction ranges use the two reference interleavings.}
 \label{tab:prediction}
 \setlength{\tabcolsep}{5pt}
 \begin{tabular}{@{}rrrrr@{}}
 \toprule
 Host GiB & Predicted prefill & Measured prefill & Predicted restore & Measured restore \\
 \midrule
 10 & 27.64--31.20 & 23.53 & 55.28--58.88 & 65.20 \\
 20 & 5.19--5.78 & 5.28 & 80.84--81.36 & 82.77 \\
 80 & 5.19--5.31 & 5.29 & 81.31--81.36 & 82.89 \\
 \bottomrule
 \end{tabular}
\end{table}

\paragraph{Changing the active pool.}
For $A=8$, the prediction additionally estimates GPU coverage from the two $A=16$ reference runs and simulates request ordering with eight active tasks; restoration serves only the coverage absent from the GPU. Table~\ref{tab:prediction-n8} lists the predictions, recorded before these runs, and the measurements.

\begin{table}[ht]
 \centering\small
 \caption{Predictions and measurements for an active pool of $A=8$. Units match Table~\ref{tab:prediction}.}
 \label{tab:prediction-n8}
 \setlength{\tabcolsep}{5pt}
 \begin{tabular}{@{}rrrrr@{}}
 \toprule
 Host GiB & Predicted prefill & Measured prefill & Predicted restore & Measured restore \\
 \midrule
 0 & 41.87--60.80 & 41.48 & 0 & 0 \\
 3 & 29.59--52.01 & 28.99 & 7.51--12.28 & 13.01 \\
 4.5 & 18.84--31.73 & 17.11 & 22.25--29.36 & 13.22 \\
 \bottomrule
 \end{tabular}
\end{table}

\paragraph{Working-set scale.}
Equation~\ref{eq:workingset} approximates a backlogged, weakly shared pool with comparable context lengths. For $A=16$, $(16-1)\times33{,}234\times24\,\mathrm{KiB}$ is 11.4\,GiB per rank; for $A=8$, $7\times33{,}234\times24\,\mathrm{KiB}$ is 5.3\,GiB. Here $\bar N=33{,}234$ is the mean prompt length of the 4,427 replayed requests; with the GPU KV capacity $K_G=343{,}408$ tokens, $\gamma_G=16\times33{,}234/343{,}408\approx1.55$ at $A=16$. This scale locates the capacity transition; the trace model accounts for the order and sizes of individual references. As tool delays, shared prefixes, and active concurrency change, the runtime updates the estimate from recent requests.

\paragraph{Tensor-parallel footprint.}
\qthree{} has 48 layers, four KV heads, and head dimension 128 \citep{qwen3CoderCard}. BF16 KV state occupies $2\times48\times4\times128\times2=98{,}304$ bytes per token before sharding. The dense Qwen2.5-Coder-32B-Instruct has 64 layers, eight KV heads, and head dimension 128 \citep{qwen25CoderCard}, or $2\times64\times8\times128\times2=262{,}144$ bytes per token. Table~\ref{tab:footprint} shows the effect of head replication, which the MoE model reaches at TP8. The GPU KV capacity that a memory fraction yields also depends on model weights and other engine allocations; Table~\ref{tab:config} lists the resulting capacities.

\begin{table}[ht]
 \centering\small
 \caption{BF16 KV footprint of the two evaluated models.}
 \label{tab:footprint}
 \setlength{\tabcolsep}{5pt}
 \begin{tabular}{@{}rrrrrrr@{}}
 \toprule
 & \multicolumn{3}{c}{\qthree{}} & \multicolumn{3}{c}{Qwen2.5-Coder-32B-Instruct} \\
 \cmidrule(lr){2-4}\cmidrule(l){5-7}
 TP ranks & Heads/rank & KiB/rank & Aggregate KiB & Heads/rank & KiB/rank & Aggregate KiB \\
 \midrule
 1 & 4 & 96 & 96 & 8 & 256 & 256 \\
 2 & 2 & 48 & 96 & 4 & 128 & 256 \\
 4 & 1 & 24 & 96 & 2 & 64 & 256 \\
 8 & 1 & 24 & 192 & 1 & 32 & 256 \\
 \bottomrule
 \end{tabular}
\end{table}

\section{Runtime Admission Details}
\label{app:runtime}

The controller runs inside LMCache's vLLM integration and decides at the request's first prefix lookup, which the vLLM scheduler issues when a scheduling step considers a waiting request for execution; write admission is thus decided within request scheduling. The controller keeps that store/skip choice for later lookups of the same request, including lookups after preemption. It changes only the storage decision: LMCache's lookup, restore, and token serialization are unchanged, and a skipped write leaves computation and decoding to the serving engine.

The LMCache worker of the first GPU publishes the tier's telemetry every second: the number of chunks registered in the tier, its chunk capacity $K_H$, and the number $e_t$ of chunks evicted in the last 60 seconds to make room for new ones. With occupancy $o_t$, the registered chunks over $K_H$, and occupancy threshold $\theta$, the tier is full and evicting when $o_t\geq\theta$ and $e_t>0$, so the pressure signal of Equation~\ref{eq:save} is
\begin{equation}
 p_t=\mathbf{1}\bigl[\widehat C_{\mathrm{reuse}}>C_H\bigr]\land\mathbf{1}\bigl[o_t\geq\theta\land e_t>0\bigr],
 \qquad
 \widehat C_{\mathrm{reuse}}=(A_t-1)\,\bar N_t\,\beta_{\mathrm{rank}},
 \label{eq:pressure-full}
\end{equation}
where $A_t$ counts the tasks with a first prefix lookup in the last 60 seconds and $\bar N_t$ is the mean prompt length of the last 256 first lookups. When no fresh report exists, as before the first publication or when a deployment does not publish telemetry, the first factor alone decides. Table~\ref{tab:controller} lists the parameters.

LMCache~0.3.12 evicts from its CPU tier only inside allocation: when a store cannot allocate a chunk buffer, the tier frees least-recently-used chunks that no pending lookup has pinned and no transfer still references, and $e_t$ counts exactly these evictions; explicit removals do not enter $e_t$. An allocation can fail while fewer chunks are registered than the tier holds, because a store allocates the buffers of all its chunks before it registers them, and buffers that transfers still reference stay allocated. The occupancy guard counts an eviction as pressure only when registered chunks fill at least $\theta$ of the tier, so evictions that capacity does not drive do not register as pressure.

\begin{table}[ht]
 \centering\small
 \caption{Write-admission parameters, fixed before the admission runs.}
 \label{tab:controller}
 \begin{tabular}{@{}ll@{}}
 \toprule
 Parameter & Value \\
 \midrule
 Cache occupancy threshold $\theta$ & 0.95 \\
 Task activity and eviction window & 60 seconds \\
 Prompt-length window & Last 256 first prefix lookups \\
 Telemetry report interval & 1 second \\
 Scheduler report-read interval & 0.5 seconds \\
 Maximum age of a fresh report & 5 seconds \\
 Write threshold $\kappa$ & 8 full chunks \\
 Host chunk size $b$ & 1,024 tokens \\
 Working-set estimate footprint & 24\,KiB/token/rank (\qthree{}, TP8) \\
 \bottomrule
 \end{tabular}
\end{table}

Deduplication operates per GPU worker before each copy. It queries the tier for the chunk keys LMCache selects for storage, groups the absent keys into contiguous runs, and passes those runs to LMCache's store routine with their original GPU KV locations. With fixed write admission at 5\,GiB, the request filter skips the writes of 1,275 requests (52,131 chunks per rank), and deduplication removes 44 chunk copies per rank (352 across the eight ranks); offload without write admission writes 79,780 chunks per rank at this budget. At 40\,GiB, the filter skips the writes of 290 requests (20,234 chunks per rank), and deduplication removes none. In the 5 and 40\,GiB runs with capacity-conditioned admission, every decision finds a fresh telemetry report. At 5\,GiB, the working-set estimate exceeds the tier for 3,994 of the 4,427 requests and telemetry reports a full, evicting tier for 4,136; pressure holds for the 3,992 requests with both, the filter skips the writes of 1,270, and deduplication removes 167 chunk copies per rank (1,336 across the eight ranks). At 40\,GiB, telemetry reports a full, evicting tier for 2,369 requests, the estimate stays below the tier for all 4,427, and no write is skipped.
At 10\,GiB, telemetry reports a full, evicting tier for 3,770 requests and the estimate exceeds the tier for 1,727; pressure holds for these 1,727, the filter skips the writes of 162, and deduplication removes 87 chunk copies per rank (696 across the eight ranks).

\paragraph{Parameter sensitivity.}
The controller makes each decision at the request's first prefix lookup from its token count $n$ and host match $h$, which the scheduler's lookup record logs for every request; the recorded decision streams therefore give $u$ exactly, and under fixed write admission, where every request is under pressure, they reproduce the filter counters above (1,275 requests and 52,131 chunks at 5\,GiB; 290 and 20,234 at 40\,GiB). New-chunk counts are bimodal (Figure~\ref{fig:sensitivity}a): at 5\,GiB, 2,814 of the 4,427 requests need at most one new full chunk, and the requests above $\kappa=8$ refill a median of 28 chunks. For $\kappa$ of 2, 4, 8, 16, and 32, the filter selects 1,468, 1,368, 1,275, 1,034, and 548 requests at 5\,GiB and declines 97.0, 96.4, 95.3, 89.6, and 68.8\% of the new-chunk writes (Figure~\ref{fig:sensitivity}b). The telemetry that the occupancy test reads is logged every 5 seconds in the three conditioned runs: all 3,049 logged reports with eviction activity (1,373, 1,059, and 617 at 5, 10, and 40\,GiB) show occupancy of at least 0.98, and 3,043 of them a full tier (Figure~\ref{fig:sensitivity}c). Any $\theta\leq0.98$ therefore yields the same pressure signal as $\theta=0.95$.

\begin{figure}[ht]
 \centering
 \includegraphics[width=\linewidth]{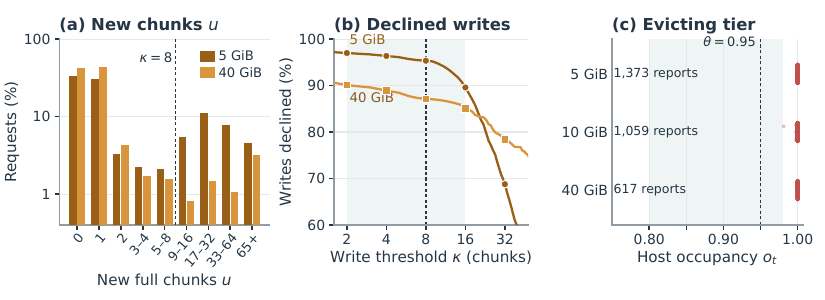}
 \caption{\textbf{The filter declines 89.6--97.0\% of new-chunk writes for $\kappa$ from 2 to 16, and every report from an evicting tier shows occupancy 0.98 or above.} (a)~New full chunks $u$ at the first prefix lookup under fixed write admission, where every request is under pressure; dashed: $\kappa=8$. (b)~Share of new-chunk writes the filter declines on the same decision streams; shading marks $\kappa$ from 2 to 16. (c)~Occupancy in every logged telemetry report with eviction activity in the conditioned runs; shading marks occupancies from 0.80 to 0.98, dashed: $\theta=0.95$.}
 \label{fig:sensitivity}
\end{figure}

\paragraph{Recency order of the host tier.}
LMCache's CPU backend orders resident chunks by recency under its default LRU policy: an insertion places a chunk at the most-recent end, a lookup hit moves the matched prefix chunks there, and an insertion that needs space evicts from the least-recent end. Proposition~\ref{prop:admission} is stated for this order.

\paragraph{Reference model.}
Let $x_1,\dots,x_T$ be a fixed stream of chunk references. The tier starts empty and holds at most $K_H$ chunks in recency order. If $x_i$ is resident, reference $i$ hits and moves $x_i$ to the most-recent position. Otherwise it misses: an admitted miss inserts $x_i$ at the most-recent position, first evicting the least-recent chunk if $K_H$ chunks are resident, and a declined miss leaves the tier unchanged. Full admission admits every miss; a restricted policy declines the misses at a set $X$ of references. Reference $i$ \emph{promotes} $x_i$ if it hits or is an admitted miss. For consecutive references $s<t$ to a chunk $y$, let $D(s,t)$ be the number of distinct chunks other than $y$ referenced strictly between $s$ and $t$, the reuse distance of Equation~\ref{eq:distance}, and let $P(s,t)$ be the number of distinct chunks other than $y$ promoted strictly between them. Every promoted chunk is referenced, so $P(s,t)\leq D(s,t)$, with equality under full admission.

\begin{lemma}
\label{lem:lru}
Under any admission policy, if reference $s$ promotes $y$ and $t$ is the next reference to $y$, then $y$ is resident at $t$ if and only if $P(s,t)<K_H$.
\end{lemma}
\begin{proof}
For $s<i\leq t$, let $Z_i$ be the set of distinct chunks other than $y$ promoted strictly between $s$ and $i$. By induction over $i$, while $y$ is resident the chunks more recent than $y$ are exactly $Z_i$. Immediately after $s$, $y$ is the most recent chunk and $Z_{s+1}=\emptyset$. A declined miss changes neither the tier nor $Z$. A promotion of a chunk in $Z$ reorders only chunks above $y$. A promotion of a chunk $z\notin Z$ moves $z$, which was resident below $y$ or absent, above $y$, and $z$ joins $Z$; if $z$ was absent and the tier full, the least-recent chunk is evicted first. While $y$ is resident, no chunk above it is least recent, so $y$ is the only chunk of $Z\cup\{y\}$ that can be evicted, and $y$ is least recent in a full tier exactly when $|Z|=K_H-1$. In that state no resident chunk lies below $y$, so the next promotion of a chunk outside $Z$ is an insertion, which evicts $y$ and raises $|Z|$ to $K_H$; $y$ cannot return before $t$, its next reference. Hence $y$ is resident at $t$ exactly when $|Z_t|=P(s,t)<K_H$.
\end{proof}

\begin{proof}[Proof of Proposition~\ref{prop:admission}]
Let reference $t$ to chunk $y$ hit under full admission. A first reference misses under every policy, so $y$ has a previous reference $s$. Under full admission every reference promotes its chunk, and Lemma~\ref{lem:lru} gives $D(s,t)=P(s,t)<K_H$. The next reference to $y$ after $s$ is $t$, at distance below $K_H$, so $s\notin X$. Under the restricted policy, $s$ therefore hits or is an admitted miss and promotes $y$, and the promotions between $s$ and $t$ satisfy $P'(s,t)\leq D(s,t)<K_H$. By Lemma~\ref{lem:lru}, $y$ is resident at $t$, and $t$ hits. The hits under full admission are thus contained in the hits under the restricted policy, for every chunk and in total.
\end{proof}
A request's host coverage is the run of hits from its first chunk, so under the restricted policy every request's consecutive prefix coverage is at least its coverage under full admission.

Conversely, let $t$ hit under the restricted policy and miss under full admission. A declined miss at $s$ would leave $y$ absent until $t$, so $s$ promotes $y$, and Lemma~\ref{lem:lru} gives $P'(s,t)<K_H\leq D(s,t)$. The added hits are exactly the references whose previous reference promotes the chunk and whose distance, counted over promoted chunks, falls below capacity once the declined insertions leave the stream; this is the load-control gain of Section~\ref{sec:method}.

\paragraph{Relation to the backlogged-pool model (Equation~\ref{eq:workingset}).}
In this model, the contexts of the other $A-1$ agents are referenced between consecutive turns of one agent, so the chunks a request writes return at the working-set scale, and the condition $D\geq K_H$ of Proposition~\ref{prop:admission} holds when $\gamma_H>1$. A large $u$ marks a request whose context the tier has already lost: when a task's earlier chunks were evicted, at least $K_H$ distinct chunks were promoted between two of its turns (Lemma~\ref{lem:lru}). Requests with $u\leq\kappa$ find all but at most $\kappa$ chunks of their context in the tier, and saving them keeps the working sets of the resident subset current.

\paragraph{Reuse distance of declined chunks.}
We measure the condition of Proposition~\ref{prop:admission} on the writes the admission runs declined. In each run's reference stream, every request references the full chunks of its prompt at its first prefix lookup, in the recorded order, with prefix-dependent keys; the declined chunks are the new full chunks of every request whose write the filter skipped. Fixed-admission decisions follow from the logged lookups, and conditioned decisions from the logged lookups and telemetry reports. For each declined chunk, $D$ counts the distinct other chunks referenced before its next reference, relative to tiers of $K_H=213$, 426, and 1,706 chunks at 5, 10, and 40\,GiB per rank (Figure~\ref{fig:declined}). Under pressure at 5\,GiB, 84.4\% (fixed) and 85.5\% (conditioned) of the declined chunks meet the condition, $D\geq K_H$ or no later reference, and their median $D$ is 1.76 and 1.71 times $K_H$. At 10\,GiB, in the transition region, 63.6\% of the chunks conditioned admission declines meet it. At 40\,GiB, where the tier holds the working set, conditioned admission declines no write, and 4.8\% of the chunks fixed admission declines meet the condition, all of them never referenced again; the others return at a median $D$ of 0.25 times $K_H$.

\begin{figure}[ht]
 \centering
 \includegraphics[width=\linewidth]{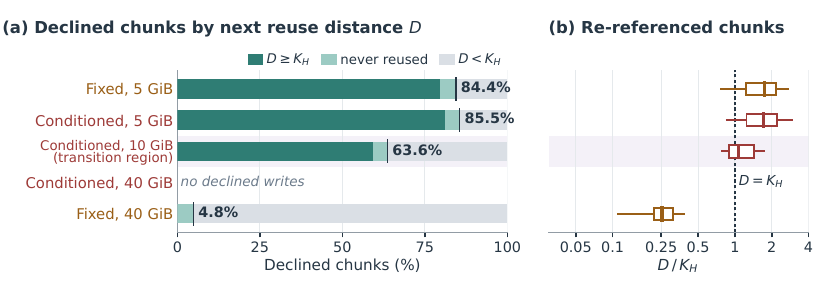}
 \caption{\textbf{Under pressure, 84.4--85.5\% of declined chunks return beyond the tier's capacity or never; at 40\,GiB, 4.8\% do.} Recorded reference streams of the admission runs at 5, 10, and 40\,GiB per rank. (a)~Declined chunks by next reuse distance $D$ relative to the tier's $K_H$ chunks; $D\geq K_H$ and no later reference together form the condition of Proposition~\ref{prop:admission}. (b)~$D/K_H$ of the declined chunks referenced again: whiskers span the 10th to 90th percentiles, boxes the quartiles, and ticks mark the median; dashed: $D=K_H$. Shading: the transition region.}
 \label{fig:declined}
\end{figure}

\FloatBarrier
\section{Hardware Deployment Results}
\label{app:hardware}

Table~\ref{tab:evidence-platform} combines completed H20 agent runs with the RTX~3090 and H800 experiments of \qthree{} and the dense Qwen2.5-Coder-32B-Instruct. Rows are grouped by timing scope (workflow wall-clock, engine window, and reported inference time), and each ratio compares offload and recomputation within one deployment and scope. Table~\ref{tab:hardware} relates these ratios to the GPU design.

\EvidencePlatformTable

Table~\ref{tab:config} lists the memory configuration of each deployment; the RTX~3090 pair uses an active pool of 16 tasks. Its sweep of the GPU memory fraction, which sets the GPU KV cache size, is shown in Figure~\ref{fig:sweep} and Table~\ref{tab:sweep}. Across the six sweep points, each percentage point of miss rate costs 2.2--2.6 minutes of wall-clock between memory fractions 0.95 and 0.80, and 6.2--17.3 minutes below 0.80.

\paragraph{Prefix caching.}
Table~\ref{tab:prefix-caching} reports \qthree{} runs with and without GPU prefix caching. Prefix caching shortens the RTX~3090 run 11.1-fold and the 2$\times$H800 run 2.19-fold.

\begin{table}[ht]
 \centering\small
 \caption{\textbf{Prefix caching shortens coding-agent runs by up to 11.1$\times$.} \qthree{}; workflow wall-clock.}
 \label{tab:prefix-caching}
 \begin{tabular}{@{}lrrr@{}}
 \toprule
 Hardware & Without prefix caching & With prefix caching & Speedup \\
 \midrule
 $8\!\times$RTX~3090 & 13\,h 35\,min & 1\,h 13\,min & \textbf{11.1}$\times$ \\
 $2\!\times$H800 & 1\,h 58\,min & 54\,min & \textbf{2.19}$\times$ \\
 \bottomrule
 \end{tabular}
\end{table} At a fixed model and precision, these deployments change both execution resources and the per-rank KV footprint.

\paragraph{Hit rates by GPU KV usage.}
Table~\ref{tab:hit-bands} reports prefix-hit medians for the \qthree{} subset pairs on RTX~3090 and two H800 GPUs at memory fraction 0.65, grouped by the GPU KV usage at which they were logged. Section~\ref{sec:deployments} compares the total medians in the three bands at 80\% usage and above.

\EvidenceHitBandTable

\begin{figure}[t]
 \centering
 \includegraphics[width=\linewidth]{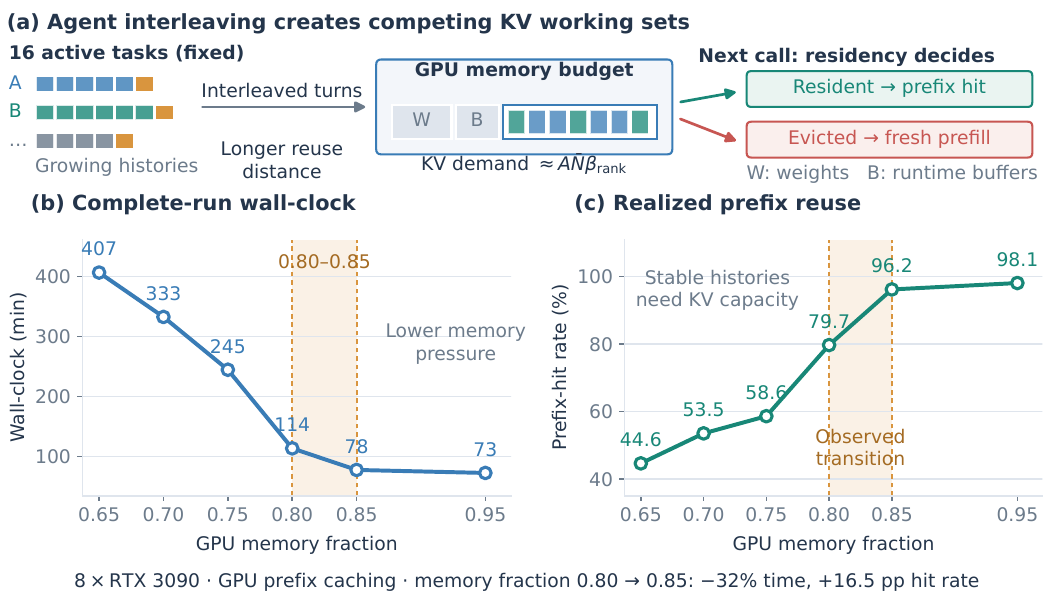}
 \caption{\textbf{The GPU tier shows the same capacity transition on RTX~3090.} The RTX~3090 sweep of the GPU memory fraction uses eight GPUs and an active pool of 16 tasks. Interleaved histories compete for resident KV state. Increasing the memory fraction from 0.80 to 0.85 lowers completion time from 114 to 78 minutes and raises the reported hit rate by 16.5 percentage points. Shading marks this measured interval.}
 \label{fig:sweep}
\end{figure}

\begin{table}[ht]
 \centering\small
 \caption{Complete RTX~3090 GPU memory-fraction sweep (\qthree{}, SWE-bench Verified subset).}
 \label{tab:sweep}
 \begin{tabular}{@{}rrr@{}}
 \toprule
 GPU memory fraction & Wall-clock (min) & Prefix hit rate \\
 \midrule
 .65 & 407 & 44.6\% \\
 .70 & 333 & 53.5\% \\
 .75 & 245 & 58.6\% \\
 .80 & 114 & 79.7\% \\
 .85 & \textbf{78} & \textbf{96.2\%} \\
 .95 & 73 & 98.1\% \\
 \bottomrule
 \end{tabular}
\end{table}

\FloatBarrier
\section{Deployment Guide}
\label{app:guide}

Table~\ref{tab:guide} turns the hardware ratio of Table~\ref{tab:hardware} and the pressure ratios $\gamma_G=A\bar N/K_G$ and $\gamma_H=\widehat C_{\mathrm{reuse}}/C_H$ of Section~\ref{sec:model}, all computable before deployment, into actions, with evidence from Section~\ref{sec:eval}.

\begin{table}[ht]
 \centering\small
 \caption{\textbf{Deployment guide.}}
 \label{tab:guide}
 \renewcommand{\arraystretch}{1.2}
 \begin{tabularx}{\linewidth}{@{}>{\raggedright\arraybackslash}p{0.27\linewidth}YY@{}}
 \toprule
 Regime & Evidence & Action \\
 \midrule
 High FLOP per host-link byte & Offload/recompute 1.08--1.87 (H800) & Favor GPU prefix caching \\
 Low FLOP per byte, $\gamma_G>1$ & 0.60 (H20, 20\,GiB); 0.91 (RTX~3090) & Offload; admit every write if $\gamma_H\leq1$ \\
 $\gamma_H>1$ & Makespan $-10.4\%$ at 5\,GiB & Conditioned admission or larger tier \\
 Context editing & Early edits rekey all later chunks & Judge edits by cache-stable length \\
 \bottomrule
 \end{tabularx}
\end{table}

\FloatBarrier
\section{Context Folding and Prefix Identity}
\label{app:folding}

\begin{figure}[ht]
 \centering
 \includegraphics[width=\linewidth]{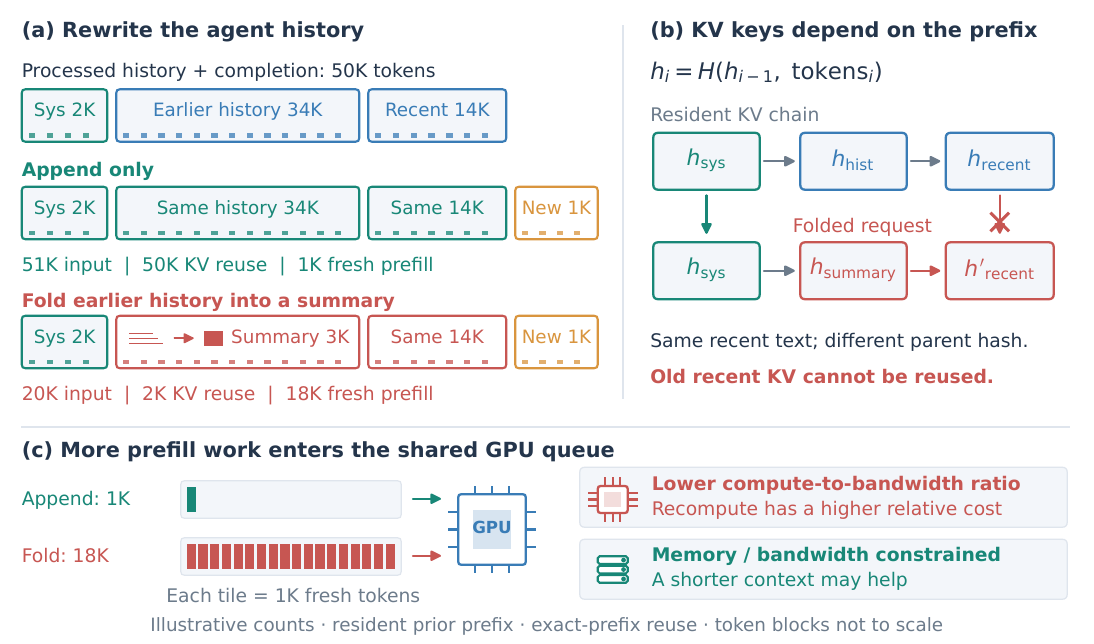}
 \caption{\textbf{Shorter input can require more fresh prefill.} In this illustrative example, retaining a processed 50K-token prefix and appending 1K tokens requires 1K tokens of fresh prefill. Folding history reduces the submitted input to 20K tokens but leaves only a 2K-token exact prefix reusable, requiring 18K fresh tokens. The example assumes resident processed history and omits block rounding. Hardware costs determine the resulting latency.}
 \label{fig:mechanism}
\end{figure}

On RTX~3090 with \qthree{} and GPU prefix caching, OpenHands condensers, which shorten the agent history before each call, change both quantities (Table~\ref{tab:folding}). Recent-event truncation and observation masking cut mean GPU KV utilization from 52.85\% to 25.57\% and 29.15\% and lower the prefix-hit rate from 96.5\% to 60.85\% and 50.80\%; LLM summarization keeps a 92.15\% hit rate at 41.12\% utilization. All three condensers lengthened per-instance inference.

\begin{table}[ht]
 \centering\small
 \caption{\textbf{Condensers that rewrite early history lower the prefix-hit rate.} \qthree{} on $8\!\times$RTX~3090 with GPU prefix caching; run means.}
 \label{tab:folding}
 \begin{tabular}{@{}lrr@{}}
 \toprule
 Condenser & GPU KV utilization & Prefix-hit rate \\
 \midrule
 None (NoOp) & 52.85\% & \textbf{96.50\%} \\
 Recent-event truncation & 25.57\% & 60.85\% \\
 Observation masking & 29.15\% & 50.80\% \\
 LLM summarization & 41.12\% & 92.15\% \\
 \bottomrule
 \end{tabular}
\end{table}

Recent-event truncation retains selected initial and recent events. Moving the retained-history boundary can edit the exact prefix from the first removed token onward. Observation masking replaces selected tool observations, while LLM summarization inserts a generated representation of earlier events. Each operation can preserve an unchanged initial prefix even when later reuse is lost. Completion tokens count as processed history because decoding has already produced their KV states.

Figure~\ref{fig:mechanism} separates submitted context length from fresh prefill. With resident processed history, fresh prefill equals the submitted length minus the cache-stable prompt length (Section~\ref{sec:background}): in the example, folding lowers the submitted length from 51K to 20K tokens and the cache-stable length from 50K to 2K, so fresh prefill rises from 1K to 18K tokens. Reuse is bounded by the cache-stable length, and a context edit is best judged by both lengths. With $N_{\mathrm{inv}}$ additional tokens recomputed once every $L_{\mathrm{cache}}$ turns, the amortized invalidation cost is $N_{\mathrm{inv}}t_{\mathrm{pf}}/L_{\mathrm{cache}}$ per turn at a fixed effective prefill cost. Both quantities depend on serialized prompts; an event threshold alone does not specify the number of model turns. This mechanism motivates accounting for prefix identity alongside token reduction.

%% file: tables/measurements.tex
\newcommand{\EvidencePlatformTable}{%
\begin{table}[ht]
\centering\small
\caption{\textbf{Offload shortens the RTX~3090 runs and lengthens every H800 run for both models.}
Both configurations enable GPU prefix caching. Rows are grouped by timing scope; each ratio (offload/recompute) compares the two within one deployment and scope.
MoE: Qwen3-Coder-30B-A3B-Instruct; dense: Qwen2.5-Coder-32B-Instruct. $^\dagger$Complete SWE-bench Verified benchmark. Bold: offload faster than recomputation.}
\label{tab:evidence-platform}
\setlength{\tabcolsep}{5pt}
\begin{tabular}{@{}llrrr@{}}
\toprule
Hardware & Model & Recompute & Offload & Ratio \\
\midrule
\multicolumn{5}{@{}l}{\emph{Workflow wall-clock}} \\
$8\!\times$RTX 3090 & MoE & 407 min & 369 min & \textbf{0.91} \\
$8\!\times$RTX 3090 & Dense & 443 min & 406 min & \textbf{0.92} \\
$8\!\times$H20 & MoE & 20.58 h & 20.80 h & 1.01 \\
\midrule
\multicolumn{5}{@{}l}{\emph{Engine window}} \\
$2\!\times$H800 & MoE & 62 min & 73 min & 1.18 \\
$2\!\times$H800 & Dense & 74 min & 91 min & 1.23 \\
\midrule
\multicolumn{5}{@{}l}{\emph{Reported inference time}} \\
$2\!\times$H800$^\dagger$ & MoE & 456 min & 492 min & 1.08 \\
$8\!\times$H800 & MoE & 53 min & 99 min & 1.87 \\
\bottomrule
\end{tabular}
\end{table}}

\newcommand{\EvidenceAllReplayTable}{%
\begin{table}[ht]
\centering\small
\caption{\textbf{Complete H20 replay measurements.}
Every completed row processes the same 4,427 calls.
Each row is one run; repeats are retained individually ($^\ast$: see Baselines and policy variants). $A$: active pool; host (CPU) budgets are per rank.
Retrieved and stored token volumes are summed over TP workers and divided by eight.
Preempt.: engine preemptions during the replay. Bold: capacity-conditioned admission.}
\label{tab:evidence-all-replay}
\setlength{\tabcolsep}{3pt}
\begin{tabular}{@{}lrrrrrrr@{}}
\toprule
Configuration & $A$ & CPU GiB & Time (min) & Prefill M & Retrieved M & Stored M & Preempt. \\
\midrule
Recompute & 16 & 0 & 211.69 & 89.66 & 0 & 0 & 2,190 \\
Recompute, repeat & 16 & 0 & 205.98 & 89.82 & 0 & 0 & 2,203 \\
Offload & 16 & 5 & 208.91 & 76.74 & 12.57 & 81.69 & 1,808 \\
Offload & 16 & 10 & 161.06 & 23.53 & 65.20 & 22.42 & 437 \\
Offload & 16 & 20 & 128.05 & 5.28 & 82.77 & 4.06 & 46 \\
Offload & 16 & 40 & 124.01 & 5.27 & 82.43 & 4.06 & 52 \\
Offload, repeat & 16 & 40 & 115.93 & 5.85 & 100.30 & 4.06 & 48 \\
Offload, repeat$^\ast$ & 16 & 40 & 123.74 & 5.29 & 81.86 & 4.06 & 44 \\
Offload & 16 & 80 & 132.82 & 5.29 & 82.89 & 4.06 & 39 \\
Fixed write admission & 16 & 40 & 161.60 & 22.79 & 65.21 & 3.06 & 537 \\
Fixed write admission & 16 & 5 & 186.39 & 48.91 & 41.25 & 2.57 & 1,183 \\
Conditioned admission & 16 & 5 & \textbf{187.08} & \textbf{49.67} & 40.41 & 3.69 & 1,233 \\
Conditioned admission & 16 & 10 & \textbf{157.41} & \textbf{20.11} & 66.88 & 5.97 & 445 \\
Conditioned admission & 16 & 40 & \textbf{124.26} & \textbf{5.28} & 82.98 & 4.06 & 50 \\
\midrule
Recompute & 8 & 0 & 159.21 & 41.48 & 0 & 0 & 843 \\
Offload & 8 & 3 & 151.11 & 28.99 & 13.01 & 50.21 & 547 \\
Offload & 8 & 4.5 & 146.75 & 17.11 & 13.22 & 26.30 & 258 \\
\bottomrule
\end{tabular}
\end{table}}

\newcommand{\EvidenceHitBandTable}{%
\begin{table}[ht]
\centering\small
\caption{\textbf{Prefix-hit medians by GPU KV usage.}
The \qthree{} RTX~3090 and two-H800 subset pairs of Table~\ref{tab:evidence-platform}, both at GPU memory fraction 0.65.
Entries are medians (\%) within each usage band. Without offload, all prefix hits are GPU hits.
With offload, the GPU, external, and total rates are separately computed medians;
the GPU and external medians need not sum to the total. Bold: net hit gain with offload at high GPU KV usage.}
\label{tab:hit-bands}
\setlength{\tabcolsep}{5pt}
\begin{tabular}{@{}llrrrr@{}}
\toprule
 & & Recompute & \multicolumn{3}{c}{Offload} \\
\cmidrule(l){4-6}
Hardware & GPU KV usage & Total & Total & GPU & External \\
\midrule
$8\!\times$RTX 3090 & $<$80\% & 62.9 & 79.4 & 62.7 & 16.6 \\
 & 80--90\% & 34.7 & \textbf{43.7} & 16.9 & 26.7 \\
 & 90--95\% & 34.4 & \textbf{44.0} & 17.1 & 27.2 \\
 & 95--100\% & 33.8 & \textbf{44.3} & 17.3 & 25.9 \\
\midrule
$2\!\times$H800 & $<$80\% & 65.7 & 98.4 & 97.5 & 0.2 \\
 & 80--90\% & 48.8 & 51.7 & 14.2 & 34.6 \\
 & 90--95\% & 47.1 & 49.0 & 14.2 & 34.4 \\
 & 95--100\% & 46.9 & 46.6 & 12.8 & 34.2 \\
\bottomrule
\end{tabular}
\end{table}}

%% file: references.bib
@inproceedings{kwon2023vllm,
  title = {Efficient Memory Management for Large Language Model Serving with {PagedAttention}},
  author = {Kwon, Woosuk and Li, Zhuohan and Zhuang, Siyuan and Sheng, Ying and Zheng, Lianmin and Yu, Cody Hao and Gonzalez, Joseph E. and Zhang, Hao and Stoica, Ion},
  booktitle = {Proceedings of the 29th Symposium on Operating Systems Principles},
  pages = {611--626},
  year = {2023},
  doi = {10.1145/3600006.3613165}
}

@inproceedings{zheng2023sglang,
  title = {{SGLang}: Efficient Execution of Structured Language Model Programs},
  author = {Zheng, Lianmin and Yin, Liangsheng and Xie, Zhiqiang and Sun, Chuyue and Huang, Jeff and Yu, Cody Hao and Cao, Shiyi and Kozyrakis, Christos and Stoica, Ion and Gonzalez, Joseph E. and Barrett, Clark and Sheng, Ying},
  booktitle = {Advances in Neural Information Processing Systems},
  volume = {37},
  year = {2024}
}

@inproceedings{wang2024openhands,
  title = {{OpenHands}: An Open Platform for {AI} Software Developers as Generalist Agents},
  author = {Wang, Xingyao and Li, Boxuan and Song, Yufan and Xu, Frank F. and Tang, Xiangru and Zhuge, Mingchen and Pan, Jiayi and Song, Yueqi and Li, Bowen and Singh, Jaskirat and Tran, Hoang H. and Li, Fuqiang and Ma, Ren and Zheng, Mingzhang and Qian, Bill and Shao, Yanjun and Muennighoff, Niklas and Zhang, Yizhe and Hui, Binyuan and Lin, Junyang and Brennan, Robert and Peng, Hao and Ji, Heng and Neubig, Graham},
  booktitle = {International Conference on Learning Representations},
  year = {2025}
}

@inproceedings{yang2024sweagent,
  title = {{SWE-agent}: Agent-Computer Interfaces Enable Automated Software Engineering},
  author = {Yang, John and Jimenez, Carlos E. and Wettig, Alexander and Lieret, Kilian and Yao, Shunyu and Narasimhan, Karthik and Press, Ofir},
  booktitle = {Advances in Neural Information Processing Systems},
  volume = {37},
  year = {2024}
}

@inproceedings{jimenez2024swebench,
  title = {{SWE-bench}: Can Language Models Resolve Real-World {GitHub} Issues?},
  author = {Jimenez, Carlos E. and Yang, John and Wettig, Alexander and Yao, Shunyu and Pei, Kexin and Press, Ofir and Narasimhan, Karthik},
  booktitle = {International Conference on Learning Representations},
  year = {2024}
}

@article{liu2025lmcache,
  title = {{LMCache}: An Efficient {KV} Cache Layer for Enterprise-Scale {LLM} Inference},
  author = {Liu, Yuhan and Yao, Jiayi and Cheng, Yihua and An, Yuwei and Chen, Xiaokun and Feng, Shaoting and Huang, Yuyang and Shen, Samuel and Zhang, Rui and Du, Kuntai and Jiang, Junchen},
  journal = {arXiv preprint arXiv:2510.09665},
  year = {2025}
}

@inproceedings{peng2024mooncake,
  title = {{Mooncake}: Trading More Storage for Less Computation --- {A} {KVCache}-centric Architecture for Serving {LLM} Chatbot},
  author = {Qin, Ruoyu and Li, Zheming and He, Weiran and Cui, Jialei and Ren, Feng and Zhang, Mingxing and Wu, Yongwei and Zheng, Weimin and Xu, Xinran},
  booktitle = {23rd USENIX Conference on File and Storage Technologies (FAST 25)},
  pages = {155--170},
  year = {2025}
}

@inproceedings{liu2024kivi,
  title = {{KIVI}: A Tuning-Free Asymmetric 2bit Quantization for {KV} Cache},
  author = {Liu, Zirui and Yuan, Jiayi and Jin, Hongye and Zhong, Shaochen and Xu, Zhaozhuo and Braverman, Vladimir and Chen, Beidi and Hu, Xia},
  booktitle = {Proceedings of the 41st International Conference on Machine Learning},
  series = {Proceedings of Machine Learning Research},
  volume = {235},
  pages = {32332--32344},
  year = {2024}
}

@inproceedings{liu2024cachegen,
  title = {{CacheGen}: {KV} Cache Compression and Streaming for Fast Large Language Model Serving},
  author = {Liu, Yuhan and Li, Hanchen and Cheng, Yihua and Ray, Siddhant and Huang, Yuyang and Zhang, Qizheng and Du, Kuntai and Yao, Jiayi and Lu, Shan and Ananthanarayanan, Ganesh and Maire, Michael and Hoffmann, Henry and Holtzman, Ari and Jiang, Junchen},
  booktitle = {Proceedings of the ACM SIGCOMM 2024 Conference},
  pages = {38--56},
  year = {2024},
  doi = {10.1145/3651890.3672274}
}

@inproceedings{pan2025kvflow,
  title = {{KVFlow}: Efficient Prefix Caching for Accelerating {LLM}-Based Multi-Agent Workflows},
  author = {Pan, Zaifeng and Patel, Ajjkumar and Shen, Yipeng and Hu, Zhengding and Guan, Yue and Li, Wan-Lu and Qin, Lianhui and Wang, Yida and Ding, Yufei},
  booktitle = {Advances in Neural Information Processing Systems},
  volume = {38},
  year = {2025}
}

@article{wang2026swepruner,
  title = {{SWE-Pruner}: Self-Adaptive Context Pruning for Coding Agents},
  author = {Wang, Yuhang and Shi, Yuling and Yang, Mo and Zhang, Rongrui and He, Shilin and Lian, Heng and Chen, Yuting and Ye, Siyu and Cai, Kai and Gu, Xiaodong},
  journal = {arXiv preprint arXiv:2601.16746},
  year = {2026}
}

@inproceedings{sun2025contextfolding,
  title = {Scaling Long-Horizon Agent via Context Folding},
  author = {Sun, Weiwei and Lu, Miao and Ling, Zhan and Liu, Kang and Yao, Xuesong and Yang, Yiming and Chen, Jiecao},
  booktitle = {Proceedings of the 43rd International Conference on Machine Learning},
  year = {2026},
  note = {Preprint titled ``Scaling Long-Horizon {LLM} Agent via Context-Folding'', arXiv:2510.11967}
}

@misc{vllmPrefixCaching,
  author = {{vLLM Team}},
  title = {{vLLM v0.13.0}: Automatic Prefix Caching},
  year = {2025},
  howpublished = {\url{https://docs.vllm.ai/en/v0.13.0/design/prefix_caching/}},
  note = {Accessed September 26, 2026}
}

@misc{qwen3CoderCard,
  author = {{Qwen Team}},
  title = {{Qwen3-Coder-30B-A3B-Instruct} Model Card},
  year = {2025},
  howpublished = {\url{https://huggingface.co/Qwen/Qwen3-Coder-30B-A3B-Instruct}}
}

@misc{qwen25CoderCard,
  author = {{Qwen Team}},
  title = {{Qwen2.5-Coder-32B-Instruct} Model Card},
  year = {2024},
  howpublished = {\url{https://huggingface.co/Qwen/Qwen2.5-Coder-32B-Instruct}}
}

@misc{topas2026,
 author={Hongqiu Ni and Han Tian and Chi Zhang and Guopeng Li and Haisheng Tan},
 title={{TOPAS}: Workflow-Aware Prefix-State Scheduling for Multi-Agent {LLM} Serving},
 year={2026}, eprint={2608.25523}, archivePrefix={arXiv},
 howpublished={arXiv preprint arXiv:2608.25523},
 url={https://arxiv.org/abs/2608.25523}
}

@misc{continuum2025,
 author={Hanchen Li and Runyuan He and Qiuyang Mang and Qizheng Zhang and Huanzhi Mao and Xiaokun Chen and Hangrui Zhou and Huanchen Zhang and Alvin Cheung and Joseph Gonzalez and Ion Stoica},
 title={{Continuum}: Efficient and Robust Multi-Turn {LLM} Agent Scheduling with {KV} Cache Time-to-Live},
 year={2026}, eprint={2511.02230}, archivePrefix={arXiv},
 howpublished={arXiv preprint arXiv:2511.02230},
 note={Version 7, revised September 2026},
 url={https://arxiv.org/abs/2511.02230v7}
}

@article{denning1968,
 author={Peter J. Denning},
 title={The Working Set Model for Program Behavior},
 journal={Communications of the ACM},
 volume={11}, number={5}, pages={323--333}, year={1968},
 doi={10.1145/363095.363141}
}

@article{denning1980,
 author={Peter J. Denning},
 title={Working Sets Past and Present},
 journal={IEEE Transactions on Software Engineering},
 volume={SE-6}, number={1}, pages={64--84}, year={1980},
 doi={10.1109/TSE.1980.230464}
}

@article{mattson1970,
 author={R. L. Mattson and J. Gecsei and D. R. Slutz and I. L. Traiger},
 title={Evaluation Techniques for Storage Hierarchies},
 journal={IBM Systems Journal},
 volume={9}, number={2}, pages={78--117}, year={1970},
 doi={10.1147/sj.92.0078}
}

@inproceedings{abhyankar2024infercept,
  title = {{InferCept}: Efficient Intercept Support for Augmented Large Language Model Inference},
  author = {Abhyankar, Reyna and He, Zijian and Srivatsa, Vikranth and Zhang, Hao and Zhang, Yiying},
  booktitle = {Proceedings of the 41st International Conference on Machine Learning},
  series = {Proceedings of Machine Learning Research},
  volume = {235},
  pages = {81--95},
  year = {2024}
}

@inproceedings{gao2024cachedattention,
  title = {Cost-Efficient Large Language Model Serving for Multi-turn Conversations with {CachedAttention}},
  author = {Gao, Bin and He, Zhuomin and Sharma, Puru and Kang, Qingxuan and Jevdjic, Djordje and Deng, Junbo and Yang, Xingkun and Yu, Zhou and Zuo, Pengfei},
  booktitle = {2024 USENIX Annual Technical Conference (USENIX ATC 24)},
  pages = {111--126},
  year = {2024}
}

@inproceedings{yu2025pensieve,
  title = {Stateful Large Language Model Serving with {Pensieve}},
  author = {Yu, Lingfan and Lin, Jinkun and Li, Jinyang},
  booktitle = {Proceedings of the Twentieth European Conference on Computer Systems},
  pages = {144--158},
  year = {2025},
  doi = {10.1145/3689031.3696086}
}

@inproceedings{gao2025hcache,
  title = {Fast State Restoration in {LLM} Serving with {HCache}},
  author = {Gao, Shiwei and Chen, Youmin and Shu, Jiwu},
  booktitle = {Proceedings of the Twentieth European Conference on Computer Systems},
  pages = {128--143},
  year = {2025},
  doi = {10.1145/3689031.3696072}
}

@inproceedings{sheng2023flexgen,
  title = {{FlexGen}: High-Throughput Generative Inference of Large Language Models with a Single {GPU}},
  author = {Sheng, Ying and Zheng, Lianmin and Yuan, Binhang and Li, Zhuohan and Ryabinin, Max and Chen, Beidi and Liang, Percy and R{\'e}, Christopher and Stoica, Ion and Zhang, Ce},
  booktitle = {Proceedings of the 40th International Conference on Machine Learning},
  series = {Proceedings of Machine Learning Research},
  volume = {202},
  pages = {31094--31116},
  year = {2023}
}

@inproceedings{lee2024infinigen,
  title = {{InfiniGen}: Efficient Generative Inference of Large Language Models with Dynamic {KV} Cache Management},
  author = {Lee, Wonbeom and Lee, Jungi and Seo, Junghwan and Sim, Jaewoong},
  booktitle = {18th USENIX Symposium on Operating Systems Design and Implementation (OSDI 24)},
  pages = {155--172},
  year = {2024}
}

@inproceedings{srivatsa2025preble,
  title = {{Preble}: Efficient Distributed Prompt Scheduling for {LLM} Serving},
  author = {Srivatsa, Vikranth and He, Zijian and Abhyankar, Reyna and Li, Dongming and Zhang, Yiying},
  booktitle = {International Conference on Learning Representations},
  year = {2025}
}

@inproceedings{luo2026agentix,
  title = {{Agentix}: An Efficient Serving Engine for {LLM} Agents as General Programs},
  author = {Luo, Michael and Shi, Xiaoxiang and Cai, Colin and Zhang, Tianjun and Wong, Justin and Wang, Yichuan and Wang, Chi and Huang, Yanping and Chen, Zhifeng and Gonzalez, Joseph E. and Stoica, Ion},
  booktitle = {23rd USENIX Symposium on Networked Systems Design and Implementation (NSDI 26)},
  pages = {2443--2459},
  year = {2026},
  note = {Preprint titled ``{Autellix}: An Efficient Serving Engine for {LLM} Agents as General Programs'', arXiv:2502.13965}
}

@inproceedings{pan2025marconi,
  title = {{Marconi}: Prefix Caching for the Era of Hybrid {LLMs}},
  author = {Pan, Rui and Wang, Zhuang and Jia, Zhen and Karakus, Can and Zancato, Luca and Dao, Tri and Wang, Yida and Netravali, Ravi},
  booktitle = {Proceedings of Machine Learning and Systems},
  volume = {7},
  year = {2025}
}

@inproceedings{wang2025kvcachewild,
  title = {{KVCache} Cache in the Wild: Characterizing and Optimizing {KVCache} Cache at a Large Cloud Provider},
  author = {Wang, Jiahao and Han, Jinbo and Wei, Xingda and Shen, Sijie and Zhang, Dingyan and Fang, Chenguang and Chen, Rong and Yu, Wenyuan and Chen, Haibo},
  booktitle = {2025 USENIX Annual Technical Conference (USENIX ATC 25)},
  pages = {465--482},
  year = {2025}
}

@misc{xia2026mori,
  author = {Tian Xia and Hanchen Li and Zhifei Li and Xiaokun Chen and Hao Kang and Yifan Qiao and Yi Xu and Ion Stoica},
  title = {Idleness is Relative: Exploiting Tool-Call Idle Windows for Offloading in Agentic Systems with {MORI}},
  year = {2026}, eprint = {2606.00866}, archivePrefix = {arXiv},
  howpublished = {arXiv preprint arXiv:2606.00866},
  url = {https://arxiv.org/abs/2606.00866}
}

@article{einziger2017tinylfu,
  title = {{TinyLFU}: A Highly Efficient Cache Admission Policy},
  author = {Einziger, Gil and Friedman, Roy and Manes, Ben},
  journal = {ACM Transactions on Storage},
  volume = {13},
  number = {4},
  pages = {35:1--35:31},
  year = {2017},
  doi = {10.1145/3149371}
}

@inproceedings{berger2017adaptsize,
  title = {{AdaptSize}: Orchestrating the Hot Object Memory Cache in a Content Delivery Network},
  author = {Berger, Daniel S. and Sitaraman, Ramesh K. and Harchol-Balter, Mor},
  booktitle = {14th USENIX Symposium on Networked Systems Design and Implementation (NSDI 17)},
  pages = {483--498},
  year = {2017}
}

@inproceedings{waldspurger2015shards,
  title = {Efficient {MRC} Construction with {SHARDS}},
  author = {Waldspurger, Carl A. and Park, Nohhyun and Garthwaite, Alexander and Ahmad, Irfan},
  booktitle = {13th USENIX Conference on File and Storage Technologies (FAST 15)},
  pages = {95--110},
  year = {2015}
}

@inproceedings{lai2001deadblock,
  title = {Dead-Block Prediction \& Dead-Block Correlating Prefetchers},
  author = {Lai, An-Chow and Fide, Cem and Falsafi, Babak},
  booktitle = {Proceedings of the 28th Annual International Symposium on Computer Architecture (ISCA)},
  pages = {144--154},
  year = {2001},
  doi = {10.1145/379240.379259}
}

@inproceedings{khan2010sampling,
  title = {Sampling Dead Block Prediction for Last-Level Caches},
  author = {Khan, Samira M. and Tian, Yingying and Jim{\'e}nez, Daniel A.},
  booktitle = {Proceedings of the 43rd Annual IEEE/ACM International Symposium on Microarchitecture (MICRO)},
  pages = {175--186},
  year = {2010},
  doi = {10.1109/MICRO.2010.24}
}

@article{johnson1999bypass,
  title = {Run-Time Cache Bypassing},
  author = {Johnson, Teresa L. and Connors, Daniel A. and Merten, Matthew C. and Hwu, Wen-mei W.},
  journal = {IEEE Transactions on Computers},
  volume = {48},
  number = {12},
  pages = {1338--1354},
  year = {1999},
  doi = {10.1109/12.817393}
}

@inproceedings{jaleel2010rrip,
  title = {High Performance Cache Replacement Using Re-Reference Interval Prediction ({RRIP})},
  author = {Jaleel, Aamer and Theobald, Kevin B. and Steely, Jr., Simon C. and Emer, Joel},
  booktitle = {Proceedings of the 37th Annual International Symposium on Computer Architecture (ISCA)},
  pages = {60--71},
  year = {2010},
  doi = {10.1145/1815961.1815971}
}

@misc{luo2025deepswe,
  title = {{DeepSWE}: Training a Fully Open-sourced, State-of-the-Art Coding Agent by Scaling {RL}},
  author = {Luo, Michael and Jain, Naman and Singh, Jaskirat and Tan, Sijun and Patel, Ameen and Wu, Qingyang and Ariyak, Alpay and Cai, Colin and Venkat, Tarun and Zhu, Shang and Athiwaratkun, Ben and Roongta, Manan and Zhang, Ce and Li, Li Erran and Popa, Raluca Ada and Sen, Koushik and Stoica, Ion},
  howpublished = {\url{https://www.together.ai/blog/deepswe}},
  note = {Agentica and Together AI blog post},
  year = {2025}
}

@article{cao2025skyrlagent,
  title = {{SkyRL-Agent}: Efficient {RL} Training for Multi-turn {LLM} Agent},
  author = {Cao, Shiyi and Li, Dacheng and Zhao, Fangzhou and Yuan, Shuo and Hegde, Sumanth R. and Chen, Connor and Ruan, Charlie and Griggs, Tyler and Liu, Shu and Tang, Eric and Liaw, Richard and Moritz, Philipp and Zaharia, Matei and Gonzalez, Joseph E. and Stoica, Ion},
  journal = {arXiv preprint arXiv:2511.16108},
  year = {2025}
}

@inproceedings{pan2025swegym,
  title = {Training Software Engineering Agents and Verifiers with {SWE-Gym}},
  author = {Pan, Jiayi and Wang, Xingyao and Neubig, Graham and Jaitly, Navdeep and Ji, Heng and Suhr, Alane and Zhang, Yizhe},
  booktitle = {Proceedings of the 42nd International Conference on Machine Learning},
  series = {Proceedings of Machine Learning Research},
  volume = {267},
  pages = {47717--47737},
  year = {2025}
}
